\documentclass[11pt]{article}
\usepackage{amsmath,amsfonts,amssymb,amsthm}
\usepackage{fullpage}
\usepackage[colorlinks]{hyperref}
\hypersetup{linkcolor=cyan,filecolor=cyan,citecolor=cyan,urlcolor=cyan}
\usepackage{enumerate}
\usepackage{cleveref}
\usepackage{float}

\usepackage{framed}
\usepackage[framemethod=tikz]{mdframed}

\usepackage[vlined, ]{algorithm2e}
\RestyleAlgo{ruled}
\LinesNumbered

\theoremstyle{plain}
\newtheorem{theorem}{Theorem}[section]
\newcommand{\BTHM}{\begin{thm}} \newcommand{\ETHM}{\end{thm}}
\newtheorem{corollary}[theorem]{Corollary}
\newcommand{\BCR}{\begin{corollary}} \newcommand{\ECR}{\end{corollary}}
\newtheorem{lemma}[theorem]{Lemma}
\newcommand{\BL}{\begin{lemma}}   \newcommand{\EL}{\end{lemma}}
\newtheorem{claim}[theorem]{Claim}
\newcommand{\BCM}{\begin{claim}}   \newcommand{\ECM}{\end{claim}}
\newtheorem{proposition}[theorem]{Proposition}
\newcommand{\BP}{\begin{proposition}}   \newcommand{\EP}{\end{proposition}}
\newtheorem{assm}[theorem]{Assumption}
\newcommand{\BASM}{\begin{assm}}   \newcommand{\EASM}{\end{assm}}

\theoremstyle{definition}
\newtheorem{definition}{Definition}[section]
\newcommand{\BD}{\begin{definition}}   \newcommand{\ED}{\end{definition}}

\newtheorem{con}[theorem]{Conjecture}
\newcommand{\BCONJ}{\begin{con}}   \newcommand{\ECONJ}{\end{con}}

\newtheorem{problem}[theorem]{Problem}
\newcommand{\BPR}{\begin{problem}}   \newcommand{\EPR}{\end{problem}}

\newenvironment{rem}{\noindent{\bf Remark:~~}}{}
\newcommand{\BREM}{\begin{rem}} \newcommand{\EREM}{\end{rem}}
\newenvironment{discussion}{\noindent{\bf Discussion:~~\\}}{}
\newcommand{\BDIS}{\begin{discussion}} \newcommand{\EDIS}{\end{discussion}}
\newtheorem{observation}{Observation}[section]

\numberwithin{equation}{section}

\newcommand{\poly}{{\rm poly}}

\newcommand{\ID}{\operatorname{ID}}
\newcommand{\EID}{\operatorname{EID}}
\newcommand{\UID}{\operatorname{UID}}

\newcommand{\LCA}{\operatorname{LCA}}
\newcommand{\depth}{\operatorname{depth}}
\newcommand{\LCALabel}{\mathsf{ANC}}

\newcommand{\Gsub}[1]{G \setminus \{ {#1} \} }
\newcommand{\conn}{\mathsf{conn}}
\newcommand{\ccid}{\mathsf{CID}}
\newcommand{\parent}{par}
\newcommand{\heavy}{h}

\newcommand{\FTLabel}[1]{ \mathsf{L}_{\mathsf{ {#1} }} }
\newcommand{\VFTLabel}[1]{ \mathsf{VL}_{\mathsf{ {#1} }} }
\newcommand{\EFTLabel}{ \mathsf{EL}}
\newcommand{\AnSet}{\mathsf{AnSet}}

\newcommand{\Null}{\mathsf{null}}
\newcommand{\nlights}{\operatorname{nl}}

\newcommand{\Sketch}{\mathsf{Sketch}}
\newcommand{\XOR}{\mathsf{XOR}}

\title{\~{O}ptimal Dual Vertex Failure Connectivity Labels}
\author{
	Merav Parter \thanks{This project is funded by the European Research Council (ERC) under the European Union’s Horizon 2020 research and innovation programme (grant agreement No. 949083), and by the Israeli Science Foundation (ISF), grant No. 2084/18.}\\
	\small Weizmann Institute \\
	\small merav.parter@weizmann.ac.il
	\and				
	Asaf Petruschka \\
	\small Weizmann Institute \\
	\small asaf.petruschka@weizmann.ac.il 
}
\date{}

\begin{document}

\maketitle

\begin{abstract}

In this paper we present succinct labeling schemes for supporting connectivity queries under vertex faults.
For a given $n$-vertex graph $G$, an $f$-VFT (resp., EFT) connectivity labeling scheme is a distributed data structure that assigns each of the graph edges and vertices a short label, such that given the labels of a vertex pair $u$ and $v$, and the labels of at most $f$ failing \emph{vertices} (resp., edges) $F$, one can determine if $u$ and $v$ are connected in $G \setminus F$. The primary complexity measure is the length of the individual labels. Since their introduction by [Courcelle, Twigg, STACS '07], FT labeling schemes have been devised only for a limited collection of graph families. A recent work [Dory and Parter, PODC 2021]  provided EFT labeling schemes for general graphs under \emph{edge} failures, leaving the vertex failure case fairly open. 

We provide the first sublinear $f$-VFT labeling schemes for $f \geq 2$ for any $n$-vertex graph. Our key result is $2$-VFT connectivity labels with $O(\log^3 n)$ bits. Our constructions are based on analyzing the structure of dual failure replacement paths on top of the well-known heavy-light tree decomposition technique of [Sleator and Tarjan, STOC 1981]. We also provide $f$-VFT labels with sub-linear length (in $|V|$) for any $f=o(\log\log n)$, that are based on a reduction to the existing EFT labels.

\end{abstract}

\newpage 

\tableofcontents
\newpage 

\section{Introduction}
Connectivity labels are among the most fundamental distributed data-structures, with a wide range of applications to graph algorithms, distributed computing and communication networks. The error-prone nature of modern day communication networks poses a demand to support a variety of logical structures and services, in the presence of vertex and edge failures. 
In this paper we study fault-tolerant (FT) connectivity labeling schemes, also known in the literature as \emph{forbidden-set} labeling. In this setting, it is required to assign each of the graph's vertices (and possibly also edges) a short name (denoted as \emph{label}), such that given the labels of a vertex pair $u$ and $v$, and the labels of a faulty-set $F$, it possible to deduce -- using no other information -- whether $u$ and $v$ are connected in $G \setminus F$. Since their introduction by Courcelle and Twigg \cite{CourcelleT07} and despite much activity revolving these topics, up until recently FT-labels have been devised only for a restricted collection of graph families. This includes graphs with bounded tree-width, planar graphs, and graphs with bounded doubling dimension \cite{CourcelleT07,abraham2012fully,AbrahamCGP16}. Hereafter, FT-labeling schemes under $f$ faults of vertices (resp., edge) are denoted by $f$-VFT labeling (resp., $f$-EFT). 

A recent work by Dory and Parter \cite{DoryP21} provided the first EFT-labeling schemes for general $n$-vertex graphs, achieving poly-logarithmic label length, independent of the number of faults $f$. For graphs with maximum degree $\Delta$, their labels immediately provide VFT-labels with $\widetilde{O}(\Delta)$ bits\footnote{By including in the label of vertex $v$ the EFT labels of all
	 edges incident to $v$.}.  The dependency on $\Delta$ is clearly undesirable, as it might be linear in $n$. This dependency can be explained by the fact that a removal of single vertex might decompose the graph into $\Theta(\Delta)$ disconnected components. The latter behavior poses a challenge for the labeling algorithm that must somehow compress the information on this large number of components into a short label.

While the $\Delta$ dependency seems to be inherent in the context of distributed vertex connectivity \cite{pritchard2011fast,parter2019small},  Baswana and Khanna and Baswana et al. \cite{khanna2010approximate,BaswanaCHR18} overcome this barrier for the single vertex fault case. Specifically, they presented a construction of distance oracles and labels, that maintain approximate distances in the presence of a single vertex fault with near linear space. This provides, in particular, $1$-VFT approximate-distance labels of polylogarithmic length.  Their constructions are based on exploiting the convenient structure of single-fault replacement paths. $1$-VFT connectivity labels of logarithmic length are easy to achieve using block-cut trees \cite{west_introduction_2000}, as discussed later on.

When turning to handling dual vertex failures,  it has been noted widely that there is
a sharp qualitative difference between a single failure and two or more failures. This one-to-two jump has been established by now for a wide-variety of fault-tolerant settings, e.g., reachability oracles \cite{Choudhary16}, distance oracles \cite{duan2009dual}, distance preservers \cite{parter2015dual,GuptaK17,ParterDISC20} and vertex-cuts \cite{HopcroftT73,BattistaT89,BattistaT96,GeorgiadisILP15}. In the lack of any $f$-VFT labeling scheme with sublinear length for any $f\geq 2$, we focus on the following natural question: 
\begin{quote}
	\emph{Is it possible to design dual vertex failure connectivity labels of $\widetilde{O}(1)$ length?}
\end{quote}

The only prior $2$-VFT labeling schemes known in the literature have been provided for \emph{directed} graphs in the special case of \emph{single-source} reachabilty by Choudhary \cite{Choudhary16}. By using the well-known tool of independent trees \cite{GeorgiadisT12}, \cite{Choudhary16} presented a construction of dual-failure \emph{single-source} reachability data structures, that also provide labels of $O(\log^3 n)$ bits. Note that in a sharp contrast to undirected connectivity that admit $O(\log n)$-bit labels, (all-pairs) reachability labels require \emph{linear} length, even in the fault-free setting \cite{DulebaGJ20}.

\paragraph{Representation of Small Vertex Cuts: Block-Cut and SPQR Trees.} 
The block-cut tree representation of a graph compactly encodes all of its single cut vertices (a.k.a. articulation points), and the remaining connected components upon the failure of each such vertex \cite{west_introduction_2000}. By associating each vertex of the original graph with a corresponding node in the block-cut tree and using standard tree labels techniques, $1$-VFT connectivity labels are easily achieved.

Moving on to dual failures, we have the similar (but more complex) SPQR-tree representation \cite{BattistaT89}, which encodes all cut pairs (i.e., vertex pair whose joint failure disconnects the graph). However, it is currently unclear to us how to utilize this structure for $2$-VFT connectivity labels.  The main issue is generalizing the vertex-node association from the block-cut tree to SPQR tree: each vertex may appear in many nodes with different `virtual edges' adjacent to it, corresponding to different cut-mates forming a cut-pair with it.

Kanevsky, Tamassia, Di Battista, and Chen \cite{KanevskyTBC91}
extended the SPQR structure to represent 3-vertex cuts.
While these representations are currently limited to cuts of size at most $3$, we hope that the approach taken in this paper can be extended to handle larger number of faults. In addition, it is arguably more distributed friendly, as it is based on basic primitives such as the heavy-light tree decomposition, which can be easily implemented in the distributed setting.  

\paragraph{On the Gap Between Edge vs. Vertex Connectivity.} Recent years have witnessed an enormous progress in our understanding of vertex cuts, from a pure graph theoretic perspective \cite{PettieY21} to many algorithmic applications  \cite{NanongkaiSY19,LiNPSY21,PettieY21,HeLW21}. Despite this exciting movement, our algorithmic toolkit for handling vertex cuts is still considerably limited compared to the counterpart setting of edge connectivity. Indeed, near-linear time sequential algorithms for edge connectivity and minimum weighted edge cuts are known for years since the celebrated result of Karger \cite{Karger93}, and its recent improvements by \cite{Ghaffari0T20,GawrychowskiMW20}. In contrast, only very recently, Forster et al. \cite{ForsterNYSY20} provided the first near linear time sequential algorithm for detecting small \emph{vertex} cuts of size $\widetilde{O}(1)$. Despite a large collection of recent groundbreaking results \cite{NanongkaiSY19,LiNPSY21,PettieY21,HeLW21}, to this date, no subquadratic-time algorithm is known for the entire connectivity regime.   

\paragraph{Additional Related Work.}
Our dual-failure vertex connectivity labels are also closely related to \emph{connectivity sensitivity oracles} \cite{DuanConnectivitySODA17,DuanP20}, that provide low-space centralized data-structure for supporting connectivity queries in presence of vertex faults. The main goal in our setting is to provide a \emph{distributed} variant of such construction, where each vertex holds only $S(n)/n$ bits of information, where $S(n)$ is the global space of the centralized data-structure. Duan and Pettie \cite{DuanConnectivitySODA17,DuanP20} provided an ingenues construction that supports multiple vertex faults in nearly optimal space of $\widetilde{O}(n)$. These constructions are built upon highly centralized building blocks, and their distributed implementation is fairly open.

\subsection{Our Contribution}\label{subsec:contribution}

We first present new constructions of $1$-VFT and $2$-VFT labeling schemes with polylogarithmic length. On a high level, our approach is based on analyzing the structure of dual failure replacement paths, and more specifically their intersection with a given heavy-light tree decomposition of (a spanning tree of) the graph. Throughout, we denote the number of graph vertices (edges) by $n$ (resp., $m$).

\paragraph{Warm-Up: $1$-VFT Connectivity Labels.} As a warm-up to our approach, we consider the single fault setting and provide a simple label description that uses only the heavy-light decomposition technique. 
\begin{theorem}[$1$-VFT Connectivity Labels]\label{thm:1vft}
For any $n$-vertex graph, there is a \emph{deterministic} $1$-VFT connectivity labeling scheme with label length of $O(\log^2 n)$ bits. The decoding algorithm takes $\poly(\log n)$ time. The labels are computed in $\widetilde{O}(m)$ randomized centralized time, or $\widetilde{O}(D)$ randomized congest rounds. 
\end{theorem}
While this construction is presented mainly to introduce our technique, it also admits an efficient distributed implementation which follows by the recent work of \cite{DistCut2022}.

\paragraph{$2$-VFT Connectivity Labels.} We then turn to consider the considerably more involved setting of supporting two vertex failures. In the literature, heavy-light tree decomposition have been proven useful mainly for handling single vertex faults, e.g. in \cite{khanna2010approximate}. The only dual-failure scheme of \cite{Choudhary16} is tailored to the \emph{single-source} setting.  By carefully analyzing dual-failure replacement paths and their interaction with the heavy-light tree decomposition of a given spanning tree, we provide deterministic labeling schemes of $O(\log^3 n)$ bits. Our main technical contribution in this paper is stated as follow:

\begin{theorem}[$2$-VFT Connectivity Labels]\label{thm:2vft}
For any $n$-vertex graph, there is a deterministic $2$-VFT connectivity labeling scheme with label length of $O(\log^3 n)$ bits. The decoding algorithm takes $\poly(\log n)$ time. The labels are computed in $\widetilde{O}(n^2)$ time.
\end{theorem}

Our dual-failure labeling scheme uses, in a complete black-box manner, single-source labels. For this purpose, we can use the $O(\log^3 n)$-bit labels of Choudhary \cite{Choudhary16}. We also provide an alternative construction that is based on the undirected tools of heavy-path tree decomposition, rather than using the tool of independent trees as in  \cite{Choudhary16}. Our single-source labels provide a somewhat improved length of $O(\log^2 n)$ bits\footnote{We note that it might also be plausible to improve the label size of \cite{Choudhary16} to $O(\log^2 n)$ bits, by reducing the size of their range-minima labels.}, but more importantly convey intuition for our all-pairs constructions.
Since our labels are built upon a single arbitrary spanning tree that can be assumed to have depth $O(D)$, we are hopeful that this approach is also more distributed-friendly. 
Specifically, as the depth of the independent trees using in \cite{Choudhary16} might be linear in $n$, their distributed computation might be too costly for the purpose of dual vertex cut computation.
Moreover, currently the tool of independent trees is limited to only two cut vertices, which also poses a barrier for extending this technique to handle multiple faults. 
In Appendix \ref{sec:single-source-dual}, we show:
\begin{lemma}\label{lem:ss-2vft_labels}
There is a \emph{single-source} 2-VFT connectivity labeling scheme with label length $O(\log^2 n)$ bits.
That is, given an $n$-vertex graph $G$ with a fixed source vertex $s$, one can label the vertices of $G$ such that given query of vertices $\langle t,x,y \rangle$ along with their labels, the connectivity of $s$ and $t$ in $G \setminus \{x,y\}$ can be inferred.
\end{lemma}

\paragraph{$f$-VFT Connectivity Labels.} Finally, we turn to consider labeling schemes in the presence of multiple vertex faults. By combining the notions of sparse vertex certificates \cite{cheriyan1993scan} with the EFT-labeling scheme of \cite{DoryP21}, in Appendix \ref{sec:sublinear} we show:
\begin{theorem}[$f$-VFT Connectivity Labels]\label{thm:fvft}
There is a $f$-VFT connectivity labeling scheme with label length $\widetilde{O}(n^{1-1/2^{f-2}})$ bits, hence of sublinear length of any $f=o(\log\log n)$. 
\end{theorem}
This for example, provides $3$-VFT labels of $\widetilde{O}(\sqrt{n})$ bits.

\subsection{Preliminaries}\label{sec:prelim}
Given a connected $n$-vertex graph $G = (V,E)$, we fix an arbitrary source vertex $s \in V$, and a spanning tree $T$ of $G$ rooted at $s$.
We assume each vertex $a$ is given a unique $O(\log n)$-bit identifier $\ID(a)$.
Let $\parent(a)$ be the parent of $a$ in $T$, $T_a$ be the subtree of $T$ rooted at $a$, and $T_a^+$ be the tree obtained from $T_a$ by connecting $\parent(a)$ to $a$.
The (unique) tree path between two vertices $a,b$ is denoted $T[a,b]$.\footnote{Note that $T[a,b]$ is a path, but $T_a$ is a subtree.}
Let $\depth(a)$ be the hop-distance of vertex $a$ from the root $s$ in $T$, i.e. the number of edges $T[s,a]$. We say that vertex $a$ is \emph{above} or \emph{higher}  (resp., \emph{below} or \emph{lower}) than vertex $b$ if $\depth(a)$ is smaller (resp., larger) than $\depth(b)$.
The vertices $a,b$ are said to be \emph{dependent} if $a$ is an ancestor of $b$ in $T$ or vice-versa, and \emph{independent} otherwise.
We denote by $\LCA (A)$ the lowest/least common ancestor in $T$ of all vertices in $A \subseteq V$.

For two paths $P,Q \subseteq G$, define the concatenation $P \circ Q$ as the path formed by concatenating $Q$ to the end of $P$.
The concatenation is well defined if for the last vertex $p_\ell$ of $P$ and the first vertex $q_f$ of $Q$ it either holds that $p_\ell = q_f$ or that $(p_\ell, q_f) \in E$.
We use the notation $P(a,b]$ for the subpath of $P$ between vertices $a$ and $b$, excluding $a$ and including $b$.
The subpaths $P[a,b)$, $P[a,b]$ and $P(a,b)$ are defined analogously.
We extend this notation for tree paths, e.g. $T(a,b]$ denotes the subpath of $T[a,b]$ obtained by omitting $a$.
When we specify $P$ as an $a$-$b$ path, we usually think of $P$ as directed from $a$ to $b$. E.g., a vertex $c \in P$ is said to be the \emph{first} having a certain property if it is the closest vertex to $a$ among all vertices of $P$ with the property.
A path $P$ \emph{avoids} a subgraph $H \subseteq G$ if they are vertex disjoint, i.e. $V(P) \cap V(H) = \emptyset$.

For a subgraph $G' \subseteq G$, let $\deg(a,G')$ be the degree of vertex $a$ in $G'$.
We denote by $\conn(a,b,G')$ the connectivity status of vertices $a$ and $b$ in $G'$, which is $1$ if $a$ and $b$ are connected in $G'$ and $0$ otherwise.
We give arbitrary unique $O(\log n)$-bit IDs to the connected components of $G'$, e.g. by taking the maximal vertex ID in each component.
We denote by $\ccid(a, G')$ the ID of the connected component containing vertex $a$ in $G'$. 
For a failure (or fault) set $F \subseteq V$, we say that two vertices $a,b$ are \emph{$F$-connected} if $\conn(a,b,G \setminus F) = 1$, and  \emph{$F$-disconnected} otherwise. In the special cases where $F = \{x\}$ or $F = \{x,y\}$ for some $x,y \in V$, we use respectively the terms $x$-connected or $xy$-connected.

\paragraph{Replacement Paths.} For a given (possibly weighted) graph $G$, vertices $a,b \in V$ and a faulty set $F \subseteq V$, the \emph{replacement path} $P_{a,b,F}$ is the shortest $a$-$b$ path in $G \setminus F$. In our context, as we are concerned with connectivity rather than in shortest-path distances, we assign weights to the graph edges for the purpose of computing replacement paths with some convenient structure w.r.t a given spanning tree $T$. Specifically, by assigning weight of $1$ to the $T$-edges, and weight $n$ to non $T$-edges, the resulting replacement paths ``walk on $T$ whenever possible''. Formally, this choice of weights ensures the following property of the replacement paths: 
for any two vertices $c,d \in P_{a,b,F}$ such that $T[c,d] \cap F = \emptyset$, $P_{a,b,F}[c,d] = T[c,d]$. Also note that these replacement paths are shortest w.r.t our weight assignment, but might not be shortest w.r.t their
 number of edges.
We may write $P_{a,b,x}$ when $F = \{ x \}$.

\paragraph{Heavy-Light Tree Decomposition.}
Our labeling schemes use the classic heavy-light tree decomposition technique introduced by Sleator and Tarjan \cite{SleatorT83}.
This is inspired by the work of Baswana and Khanna \cite{khanna2010approximate} applying this technique in the fault-tolerant setting.
The \emph{heavy child} of a non-leaf vertex $a$ in $T$, denoted $\heavy(a)$, is the child $b$ of $a$ that maximizes the number of vertices in its subtree $T_b$ (ties are broken arbitrarily in a consistent manner.).
A vertex is called \emph{heavy} if it is the heavy child of its parent, and \emph{light} otherwise. 
A tree edge in $T$ is called \emph{heavy} if it connects a vertex to its heavy child, and \emph{light} otherwise. The set of heavy edges induces a collection of tree paths, which we call \emph{heavy paths}. 
Let $a,b \in V$ such that $a$ is a strict ancestor of $b$, and let $a'$ be the child of $a$ on $T[a,b]$.
Then $a$ is called a \emph{heavy ancestor} of $b$ if $a'$ is heavy, or a \emph{light ancestor} of $b$ if $a'$ is light. Note that a heavy ancestor of $b$ need not be a heavy vertex itself, and similarly for light ancestors.
We observe that if $b$ is a light child of $a$, then $T_b$ contains at most half of the vertices in $T_a$. Consequently, we have:
\begin{observation}\label{obs:heavy-light}
	Any root-to-leaf path in $T$ contains only $O(\log n)$ light vertices and edges.
\end{observation}
Our labeling schemes are based on identifying for each vertex $a$ a small number of $\emph{interesting}$ vertices, selected based on the heavy-light decomposition.
\begin{definition}\label{def:intersets}
	The \emph{interesting set} of a vertex $a$ is defined to be $I(a) = \{b \in T[s,a] \mid b \text{ is light}\} \cup \{ \heavy(a) \}$ (where $\{\heavy(a)\}$ is interpreted as the empty set if $a$ is a leaf).
	That is, $I(a)$ consists of all the light vertices on $T[s,a]$, along with the heavy child of $a$ (if it exists).
	The \emph{upper-interesting set} of $a$ is defined to be $I^{\uparrow} (a) = \{\parent(b) \mid b \in I(a) \} \cup \{a\}$.
	That is, $I^{\uparrow} (a)$ consists of all the light ancestors of $a$ and $a$ itself.
\end{definition}
We make extensive use of the following useful properties of interesting sets, which are immediate to prove.
\begin{lemma}\label{lem:instersting_sets}
	For any $a\ \in V$, $|I(a)| = O(\log n)$ and $|I^{\uparrow}(a)| = O(\log n)$.
\end{lemma}

\begin{lemma}\label{lem:child_on_path}
	Let $a,b \in V$ such that $a \in T[s,b]$.
	
	(1) If $a \neq b$, then for the child $a'$ of $a$ on $T[a,b]$ it holds that $a' \in I(a) \cup I(b)$.
	
	(2) If $a \notin I^{\uparrow} (b)$, then $a \neq b$ and the child of $a$ on $T[a,b]$ is $\heavy(a)$.	
\end{lemma}

\paragraph{Extended Vertex IDs.}
To avoid cumbersome definitions in our labels, it is convenient to augment the vertex IDs with additional $O(\log n)$ bits of information, resulting in \emph{extended IDs}.
The main ingredient is \emph{ancestry labels} \cite{kannan1992implicit}: these are $O(\log n)$-bit labels $\LCALabel_T (a)$ for each vertex $a$, such that given $\LCALabel_T (a)$ and $\LCALabel_T (b)$ one can infer whether $a$ is an ancestor of $b$ in $T$.
The extended ID of a vertex $a$ is\footnote{If $a$ is a leaf, we simply omit from $\EID(a)$ the information regarding $\heavy(a)$.}
$
\EID(a) = \left[ \ID(a), \LCALabel_{T}(a), \ID(\heavy(a)), \LCALabel_{T}(\heavy(a)) \right]
$.
Thus, given $\EID(a)$ and $\EID(b)$, one can determine whether $a$ is an ancestor of $b$ in $T$, and moreover, whether it is a light or a heavy ancestor.
We will not explicitly refer to the extended IDs, but rather use them as follows:
\begin{itemize}
	\item The label of any vertex $a$ always (implicitly) stores $\EID(a)$.
	\item Whenever a label stores a given vertex $a$, it additionally stores $\EID(a)$.
\end{itemize}
This enables us to assume throughout that we can always determine the (heavy or light) ancestry relations of the vertices at play.

\section{Single Failure Connectivity Labels}\label{sec:single-fault}

In this section we warm-up by considering the single failure case of Theorem \ref{thm:1vft}.
\begin{algorithm}[!h]
	\caption{Construction of label $\FTLabel{1F} (a)$ for vertex $a$}\label{alg:1f-label}
	\For{each $b' \in I(a)$ with $\parent(b')=b$}{
		\textbf{store} vertices $b,b'$ and the values $\conn (s, b', \Gsub{b})$, $\ccid(b', \Gsub{b})$\;
	}
\end{algorithm}

\noindent The key observation for decoding is:
\begin{claim}\label{claim:1vft}
	Given $\FTLabel{1F} (w)$ and  $\FTLabel{1F} (x)$, one can determine the $x$-connectivity of $w$ and $s$, and also find $\ccid(w, \Gsub{x})$ in case $w,s$ are $x$-disconnected.
\end{claim}
\begin{proof}
	If $w$ is not a descendant\footnote{This is checked using extended IDs $\EID(w)$ and $\EID(x)$.} of $x$, then $T[s,w]$ is failure-free, so $w$ and $s$ are $x$-connected and we are done.
	Assume now that $w$ is a descendant of $x$, and let $x'$ be the child of $x$ on $T[x,w]$.
	Then $T[x',w]$ is failure-free, hence $x'$ and $w$ are $x$-connected.
	Therefore, it suffices to determine the values $\conn(s,x',\Gsub{x})$ and $\ccid(x', \Gsub{x})$.
	Lemma \ref{lem:child_on_path}(1) guarantees that $x' \in I(w) \cup I(x)$, hence the required values are stored either in $\FTLabel{1F} (w)$ or in $\FTLabel{1F} (x)$ (by setting $b=x$ and $b'=x'$).
\end{proof}
Given $\FTLabel{1F} (u)$, $\FTLabel{1F}  (v)$ and $\FTLabel{1F}  (x)$, we determine the $x$-connectivity of $u,v$ as follows.
We apply Claim \ref{claim:1vft} twice, with $w=u$ and with $w=v$.
If we find the component IDs of both $u$ and $v$ in $\Gsub{x}$, 
we compare them and answer accordingly. 
However, if this is not the case, then we must discover that one of $u,v$ is $x$-connected to $s$, so we should answer affirmatively iff the other is $x$-connected to $s$. This completes the decoding algorithm of Theorem \ref{thm:1vft}. The preprocessing time analysis is deferred Appendix \ref{sec:missing-proofs}.

\def\APPENDTIMEONEF{
\paragraph{Preprocessing Time for $1$-VFT Labels (Theorem \ref{thm:1vft}).}
We start by computing a $2$-vertex connectivity certificates $G' \subseteq G$ such that $|E(G')|=O(n)$, which can be done in $O(m)$ time by \cite{cheriyan1993scan}. We then apply the labeling algorithm on the graph $G'$. It is easy to see that each label $\FTLabel{1F}  (v)$ can be computed by applying $\widetilde{O}(1)$ connectivity algorithms, hence taking $\widetilde{O}(n)$ time. This completes the deterministic construction of Theorem \ref{thm:1vft}. The linear-time randomized computation, in fact, follows by our distributed computation of these labels which are based on \cite{DistCut2022}, is deferred to Appendix \ref{sec:distributed-comp}. 
}

\section{Dual Failure Connectivity Labels}\label{sec:dual-fault}

\subsection{Technical Overview}\label{sec:techniques}
In the following we provide high-level intuition for our main technical contribution of dual failure connectivity labels. Throughout, the query is given by the tuple $\langle u,v, x,y\rangle$, where $x,y$ are the vertex faults. Recall that our construction is based on some underlying spanning tree $T$ rooted at some vertex $s$ (that we treat as the \emph{source}). Similarly to the $1$-VFT construction, we compute the heavy-light tree decomposition of $T$, which classifies the tree edges into heavy and light. 


We distinguish between two structural cases depending on the locations of the two faults, $x$ and $y$. The first case which we call \emph{dependent} handles the setting where $x$ and $y$ have ancestry/descendant relations. The second \emph{independent} case assumes that $x$ and $y$ are not dependent, i.e., $\LCA(x,y)\notin \{x,y\}$. 

Our starting observation is that by using single-source $2$-VFT labels in a black-box manner, we may restrict our attention to the hard case where the source $s$ is $xy$-disconnected from both $u$ and $v$. 
Quite surprisingly, this assumption yields meaningful restrictions on the structure of key configurations, as will be demonstrated shortly.

\paragraph{Dependent Failures.}
To gain intuition, we delve into two extremes: the easy \emph{all-light} case where $u,v,y$ are all \emph{light} descendants of $x$, and the difficult \emph{all-heavy} case where they are all \emph{heavy} descendants of $x$.
Consider first the easy all-light case.
As every vertex has only $O(\log n)$ light ancestors, each of $u,v,y$ has the budget to prepare by storing its $1$-VFT label w.r.t the graph $\Gsub{x}$. Then, for decoding, we simply answer the single-failure query $\langle u,v,y \rangle$ in $\Gsub{x}$.

We turn to consider the all-heavy case, which turns out to be an important core configuration.
Here, we no longer have the budget to prepare for each possible failing $x$, and a more careful inspection is required.
The interesting case is when $y \in T[u,v]$.
We further focus in this overview on the following instructive situation: $y$ is not an ancestor of $v$, but is a \emph{heavy} ancestor of $u$.
It is then sufficient to determine the $xy$-connectivity of $\heavy(y)$ and $\parent(y)$. See Figure \ref{fig:tech_overview} (left).
\begin{figure}
	\centering
	\includegraphics[scale=0.5]{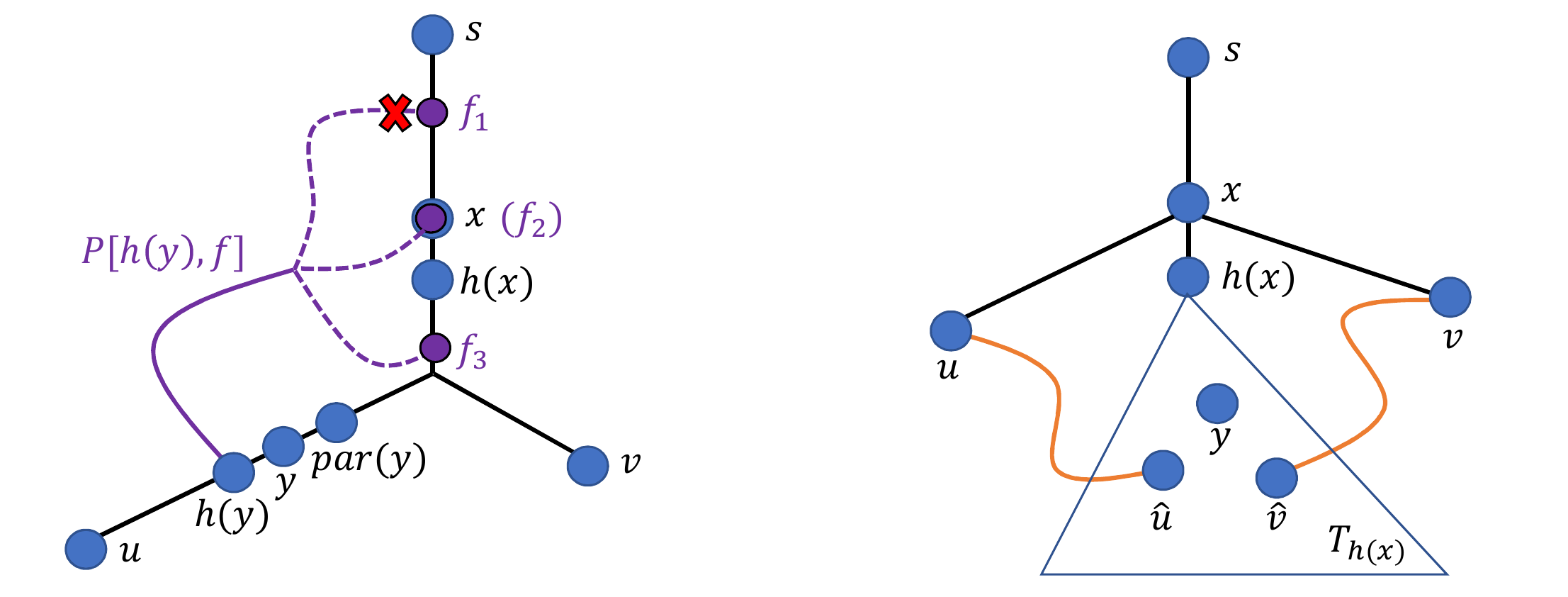}
	\caption{
		Left: Illustration of the all-heavy configuration. 
		Letting $P = P_{\heavy(y),\parent(y),y}$, the purple path represents the prefix $P[\heavy(y),f]$ of $P$ until the first time it hits $T[s,\parent(y)]$.
		Vertices $f_1, f_2, f_3$ correspond to different options for the location of $f$: above $x$, equals $x$, or below $x$. The $f_1$ option is marked X as it is excluded by our analysis. 
		Right: Illustration of the reduction to the all-heavy case. The analog vertices $\widehat{u}, \widehat{v}$ are chosen from $\AnSet(u,x), \AnSet(v,x)$, respectively.
	}
	\label{fig:tech_overview}
\end{figure}
Naturally, $y$ is most suited to prepare in advance for this situation, as follows.
Let $P = P_{\heavy(y), \parent(y), y}$  be the $\heavy(y)$-$\parent(y)$ replacement path avoiding $y$, and let $f \in P$ be the \emph{first} vertex (i.e., closest to $\heavy(y)$) from $T[s, \parent(y)]$.
Surprisingly, it suffices for the labeling algorithm to include in the label of $y$ the identity of $f$, along with a single bit representing the connectivity of $\heavy(y)$ and $\parent(y)$ in $G \setminus \{f,y\}$.

This limited amount of information turns out to be sufficient thanks to the useful structures of the replacement paths. Since $x$ is an ancestor of $y$, we need to consider the possible locations of $x$ within $T[s, \parent(y)]$.
The key observation is that $x$ cannot lie below $f$, i.e. in $T(f,\parent(u)]$: otherwise, $T[u, \heavy(y)] \circ P[\heavy(y), f] \circ T[f, s]$ is a $u$-$s$ path avoiding $x,y$, which we assume does not exist!
Now, if $x$ is above $f$, i.e. in $T[s,f)$, then $P$ is fault-free, so we determine that $\heavy(y),\parent(y)$ are $xy$-connected. If $x = f$, we simply have the answer stored explicitly by the label of $y$.
The complete solution for the all-heavy case is of a similar flavor, albeit somewhat more involved.

We then handle the general dependent failures case by reducing to the all-heavy configuration, which we next describe in broad strokes.
First, the case where $y$ is a light descendant of $x$ is handled directly using $1$-VFT labels, in a similar manner to the all-light case.
In the remaining case where $y \in T_{\heavy(x)}$, a challenge arises when (at least) one of $u,v$, say $u$, is a light descendant of $x$. We exploit the fact that the label of $u$ has the \emph{budget} to prepare for light ancestors, and store in this label a small and carefully chosen set of vertices from $T_{\heavy(x)}$, called the \emph{analog set} $\AnSet(u,x)$. Our decoding algorithm in this case replaces the given $\langle u, v, x, y \rangle$ query with an \emph{analogous} all-heavy query $\langle \widehat{u}, \widehat{v}, x, y \rangle$ for some $\widehat{u} \in \AnSet(u,x)$ and $\widehat{v} \in \AnSet(v,x)$. See Figure \ref{fig:tech_overview} (right). The reduction's correctness is guaranteed by the definition of analog sets.

\paragraph{Independent Failures.}
Our intuition comes from our solution to the  \emph{single-source} independent-failures case, described in Appendix \ref{sec:single-source-dual-ind}.
As we assume that both $u,v$ are $xy$-disconnected from $s$, we know that the corresponding decoding algorithm \emph{rejects} both queries $\langle u, x, y \rangle$ and  $\langle v, x, y \rangle$.
Rejection instances can be of two types: \emph{explicit reject} or \emph{implicit reject}.

If $\langle u,x,y \rangle$ is an explicit reject instance, then the algorithm rejects by tracking down an explicit bit stored in one of the labels of $u,x,y$, and returning it. This bit is of the form $\conn(s, \widetilde{u}, \Gsub{x,y})$ for some vertex $\widetilde{u}$ which is $xy$-connected to $u$.
So, in explicit reject instances, one of the vertices $u,x,y$  has \emph{prepared in advance} by storing this bit.
In contrast, if it is an implicit reject instance, then the algorithm detects \emph{at query time} that $\langle u,x,y \rangle$ match a specific, highly-structured fatal configuration leading to rejection.
Specifically, this configuration implies that $u$ is $xy$-connected to both $\heavy(x)$ and $\heavy(y)$.
Thus, when reaching implicit rejection, we can infer useful structural information.

Getting back to our original query $\langle u, v, x, y \rangle$, the idea is to handle all of the four possible combinations of implicit or explicit rejects for $\langle u, x, y \rangle$ or $ \langle v, x, y \rangle$.
Our approach is then based on augmenting the single-source $2$-VFT labels in order to provide the decoding algorithm with a richer information in the explicit rejection cases.

The presented formal solution distills the relevant properties of the corresponding (augmented) single-source labels and decoding algorithm. This approach has the advantage of having a succinct, clear and stand-alone presentation which does not require any prior knowledge of our single-source solution, but might hide some of  the aforementioned intuition.

\paragraph{Setting Up the Basic Assumptions.} 
We now precisely describe the basic assumptions that we enforce as preliminary step.
These are:
\begin{enumerate}[(C1)]
	\item $u$ and $v$ are both $x$-connected and $y$-connected.
	\item Both $u$ and $v$ are $xy$-disconnected from the source $s$.
\end{enumerate}
To verify condition (C1), we augment the $2$-VFT label of each vertex with its $1$-VFT label from Theorem \ref{thm:1vft}. If (C1) is not satisfied, then clearly $u,v$ are $xy$-disconnected, so we are done.
For (C2), we further augment the labels with the corresponding single-source $2$-VFT labels of \cite{Choudhary16} (or alternatively, our own such labels of Lemma \ref{lem:ss-2vft_labels}).
Using them we check whether $u$ and $v$ are $xy$-connected to $s$.
If both answers are affirmative, then $u,v$ are $xy$-connected.
If the answers are different, then $u,v$ are $xy$-disconnected.
Hence, the only non-trivial situation is when (C2) holds.
As the $1$-VFT and single-source $2$-VFT labels consume only $O(\log^3 n)$ bits each, the above mentioned augmentations are within our budget.

The following sections present our $2$-VFT connectivity labeling scheme in detail: Section \ref{sec:dual-ind} handles the independent-failures case, and Section \ref{sec:dual-dep} considers the dependent-failure case.
The final label is obtained by adding the sublables provided in each of these sections.
We note that by using the extended IDs of the vertices, it is easy for the decoding algorithm to detect which of the cases fits the given $\langle u, v, x, y \rangle$ query.

\subsection{Two Failures are Independent}\label{sec:dual-ind}

The independence of the failures allows us to enforce a stronger version of condition (C1):
\begin{enumerate}[(C3)]
	\item $u$, $v$ and $s$ are all $x$-connected and $y$-connected.
\end{enumerate}
Condition (C3) is verified using the 1-VFT labels of $u,v,x,y$ and $s$.
As the label of $s$ is not given to us, we just augment the label of every vertex also with the $1$-VFT label of $s$.
If (C3) fails, we are done by the following claim.
Missing proofs are deferred to Appendix \ref{sec:missing-proofs}.

\begin{claim}\label{claim:c3}
	If condition (C3) does not hold, then $u$ and $v$ are $xy$-connected.
\end{claim}
\def\CTHREEPROOF{
\begin{proof}
	Assume (C3) does not hold. By (C1), this can happen only if one of $u,v$, say $u$, is disconnected from $s$ under one of the failures, say $x$. Namely, $u,s$ are $x$-disconnected.
	Now, (C1) also ensures that there is a $u$-$v$ path $P$ avoiding $x$. 
	We assert that $P$ also avoids $y$, which completes the proof.
	Assume otherwise, and consider the $u$-$s$ path $P' = P[u,y] \circ T[y,s]$.
	By the independence of $x,y$ we have that $x \notin P'$, which contradicts the fact that $u,s$ are $x$-disconnected. 
\end{proof}
}\CTHREEPROOF

\noindent Our general strategy is to design labels $\FTLabel{P} (a)$ for each vertex $a$ that have following property:
\begin{enumerate}[(P)]
\item For any $\langle u,v,x,y \rangle$ with independent failures\footnote{Which satisfy conditions (C1), (C2) and (C3).} $x,y$, there exists\footnote{At least one of $h(x),h(y)$ exists: else, $x,y$ are leaves, so  $T \setminus \{x, y\}$ spans $\Gsub{x,y}$, contradicting (C2).} $z \in \{\heavy(x), \heavy(y)\}$ such that given the label $\FTLabel{P} (w)$ of any $w \in \{ u, v \}$ and the labels $\FTLabel{P}(x), \FTLabel{P}(y)$, one can infer the $xy$-connectivity of $w,z$, and also find $\ccid(w, \Gsub{x,y})$ in case $w,z$ are $xy$-disconnected.
\end{enumerate}
This suffices to determine the $xy$-connectivity of $u,v$ by the following lemma:

\begin{lemma}\label{lem:property-P}
	Given the $\FTLabel{P}$ labels of $u,v,x,y$, one can determine the $xy$-connectivity of $u,v$.
\end{lemma}

\begin{proof}
	We apply property (P) twice, for $w=u$ and for $w=v$.
	If we find the component IDs of both $u$ and $v$ in $\Gsub{x,y}$ we just compare them and answer accordingly.
	However, if this does not happen, then we must discover that one of $u, v$ is $xy$-connected to $z$, so we should answer affirmatively iff the other is also $xy$-connected to $z$.
\end{proof}

In order to preserve the symmetry between the independent failures $x$ and $y$ (which we prefer to break in a more favorable manner in our subsequent technical arguments), we do not explicitly specify, at this point, the identity of $z \in \{\heavy(x),\heavy(y)\}$. It may be useful for the reader to think of $z$ as chosen \emph{adversarially} from $\{\heavy(x), \heavy(y)\}$, and our decoding algorithm handles each of the two possible selections. Alternatively, this can be put as follows: our labeling scheme will guarantee property (P) for $z = \heavy(x)$ \emph{and} for $z = \heavy(y)$ (if both heavy children exist).

\paragraph{Construction of $\FTLabel{P}$ Labels.}
We start with a useful property of single-fault replacement paths.
For a vertex $a \in V$ with an $s$-$a$ replacement path $P = P_{s,a,\parent(a)}$, let $\ell_{a} \in P$ be the last (closest to $a$) vertex in $T \setminus T_{\parent(a)}$.
\begin{observation}\label{obs:the_ell_vertex}
	$P_{s,a,\parent(a)} = T[s, \ell_a] \circ Q$ where $Q \subseteq T_{\parent(a)}$.
\end{observation}
We are now ready to define the $\FTLabel{P}$ labels. These are constructed by Algorithm \ref{alg:P-label}.
The label length of $O(\log^3 n)$ bits follows by Lemma \ref{lem:instersting_sets}.

$\quad$\newline
\begin{algorithm}[H]
	\caption{Construction of label $\FTLabel{P} (a)$ for vertex $a$}\label{alg:P-label}
	\For{each $b' \in I(a)$ with $\parent(b')=b$}{
		\textbf{store} vertices $b, b', \ell_{b'}$\;
		\For{each $c \in I^{\uparrow}(\ell_{b'})$}{
			\textbf{store} vertex $c$\;
			\textbf{store} $\ccid(b', \Gsub{b,c})$, $\conn(b', \heavy(b), \Gsub{b,c})$, $\conn(b', \heavy(c), \Gsub{b,c})$\;
		}
	}
\end{algorithm}

\paragraph{Decoding Algorithm for Property (P)}
Our goal is to show that given $\FTLabel{P} (w)$ for $w \in \{u,v\}$ and $\FTLabel{P} (x), \FTLabel{P} (y)$ we can indeed satisfy the promise of (P); namely, determine the $xy$-connectivity of $w,z$, and in case they are $xy$-disconnected also report $\ccid(w, \Gsub{x,y})$.

One of the failures, say $x$, must be an ancestor of $w$ in $T$.
Otherwise, $w$ would have been connected to $s$ in $\Gsub{x,y}$, contradicting (C2).
Denote by $x'$ the child of $x$ on $T[x,w]$.
The independence of $x,y$ guarantees that $T[x',w]$ is fault-free, hence $x',w$ are $xy$-connected.

It now follows from (C2) that $x',s$ are $xy$-disconnected, and from (C3) that $x',s$ are $x$-connected.
The latter ensures that $\ell_{x'}$ is well-defined.
By Observation \ref{obs:the_ell_vertex}, $P_{s,x',x} = T[s,\ell_{x'}] \circ Q$ where $Q \subseteq T_x$, so $y \notin Q$.
On the other hand, it cannot be $P_{s,x',x}$ entirely avoids $y$, as we have already established that $s,x'$ are $xy$-disconnected.
It follows that $y \in T[s,\ell_{x'}]$. See illustration in Figure \ref{fig:ind}.
\begin{figure}
	\centering
	\includegraphics[scale=0.5]{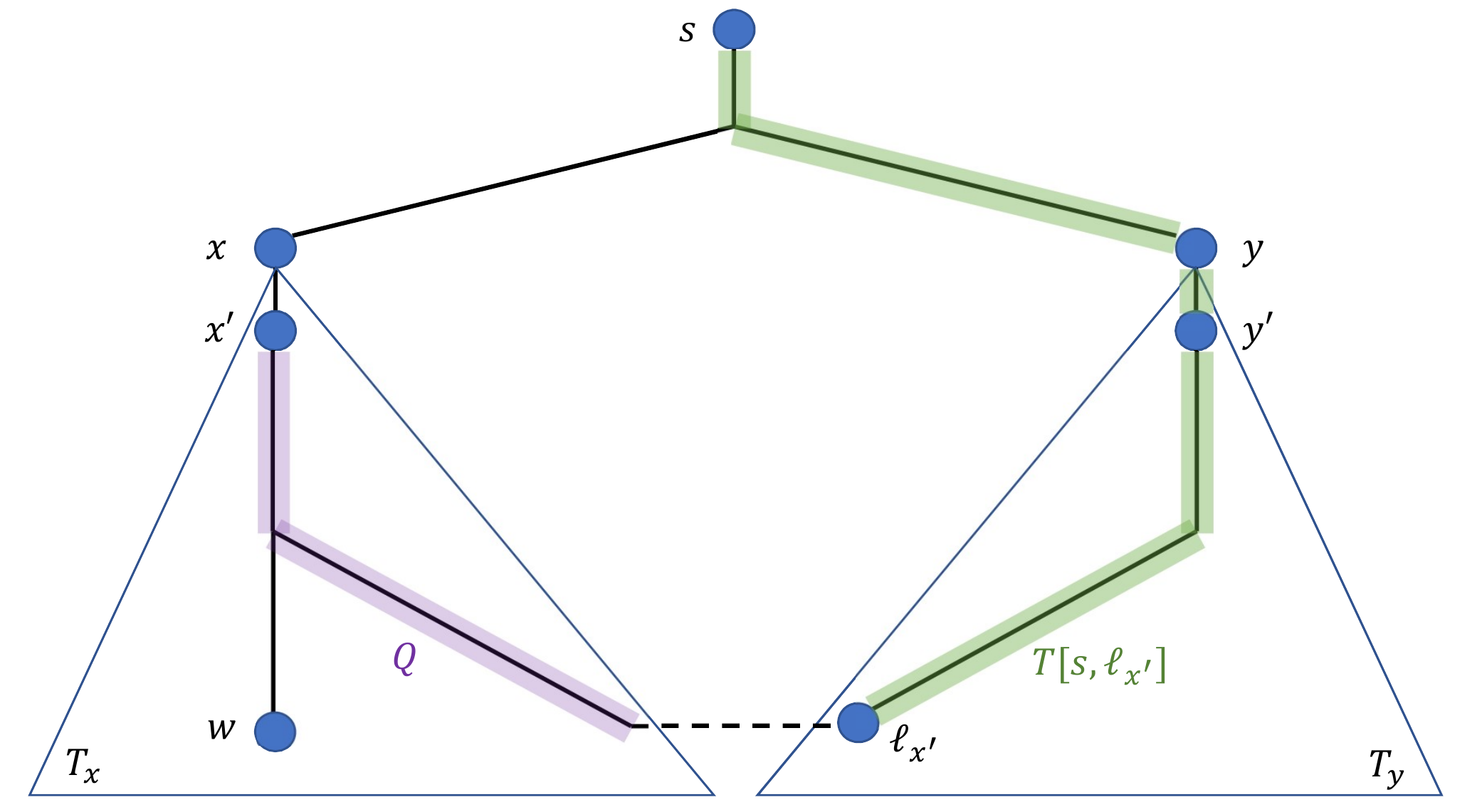}
	\caption{
		Illustration of the decoding algorithm for the independent failures case. The path $T[s,\ell_{x'}]$ is shown in green (right), and the path $Q$ is shown in purple (left). The concatenation $T[s,\ell_{x'}] \circ Q$ forms the replacement path $P_{s,x',x}$.
		The vertex $y'$ is the child of $y$ on $T[s, \ell_{x'}]$. The case where $y \notin I^{\uparrow} (\ell_{x'})$ corresponds to $y'=h(y)$.
	}
	\label{fig:ind}
\end{figure}
Note that $x' \in I(w) \cup I(x)$ by Lemma \ref{lem:child_on_path}, hence the triplet of vertices $b = x$, $b' = x'$ and $\ell_{b'} = \ell_{x'}$ is stored either in $\FTLabel{P} (w)$ or in $\FTLabel{P} (x)$.
We next distinguish between two cases, depending on whether $y$ belongs to the upper-interesting set of $\ell_{x'}$.

\smallskip\noindent\textbf{Case 1: $y \in I^{\uparrow} (\ell_{x'})$.}
 Then the following are also specified in the last label (with $c=y$):
 \begin{equation*}
	\ccid(x', \Gsub{x,y}), \;
	\conn(x', \heavy(x), \Gsub{x,y}), \;
	\conn(x', \heavy(y), \Gsub{x,y}).
 \end{equation*} 
Since $x',w$ are $xy$-connected, we can replace $x'$ by $w$ in the three values above, and reporting them guarantees property (P).

\smallskip\noindent\textbf{Case 2: $y \notin I^{\uparrow} (\ell_{x'})$.}
Then the child of $y$ on $T[y,\ell_{x'}]$ is $\heavy(y)$ by Lemma \ref{lem:child_on_path}.
Thus $T[\heavy(y), \ell_{x'}] \circ Q$ is a $\heavy(y)$-$x'$ path avoiding both $x$ and $y$.
Therefore, as $w$ is $xy$-connected to $x'$, it is also $xy$-connected to $\heavy(y)$.
So, if $z=\heavy(y)$ we are done.
However, if $z = \heavy(x)$, we recover by simply repeating the algorithm when $y,\heavy(y),x$ play the respective roles of $x,x',y$.
The triplet $y, \heavy(y), \ell_{\heavy(y)}$ is stored $\FTLabel{P} (y)$ since $h(y) \in I(y)$.
If $x \in I^{\uparrow}(\ell_{\heavy(y)})$, this label also specifies
\begin{equation*}
	\ccid(\heavy(y), \Gsub{x,y}), \;
	\conn(\heavy(y), \heavy(x), \Gsub{x,y}), \;
\end{equation*} 
and since $w,\heavy(y)$ are $xy$-connected we can replace $\heavy(y)$ by $w$ in these values, and report them to guarantee property (P).
Otherwise, we deduce (by a symmetric argument) that $w$ and $\heavy(x) = z$ are $xy$-connected, so we are done.
This concludes the decoding algorithm for property (P).
Finally, by Lemma \ref{lem:property-P}, we obtain:
\begin{lemma}\label{lem:ind-2vft}
	 There are $O(\log^3 n)$-bit labels $\FTLabel{ind}$ supporting the independent failures case.
\end{lemma}

\subsection{Two Failures are Dependent}\label{sec:dual-dep}
In this section, we consider the complementary case where $x$ and $y$ are dependent.
As previously discussed, our strategy is based on reducing the general dependent-failures case to the  well-structured configuration of the all-heavy case:
\begin{definition}
	A query of vertices $\langle u,v,x,y \rangle$ is said to be \emph{all-heavy (AH)} if $u,v,y \in T_{\heavy(x)}$.
\end{definition}
We first handle this configuration in Section \ref{sec:AH} by defining sub-labels $\FTLabel{AH}$ that are tailored to handle it.
Then, Section \ref{sec:reduction_to_AH} considers the general dependent-failures case.

\subsubsection{The All-Heavy (AH) Case}\label{sec:AH}

\paragraph{Construction of $\FTLabel{AH}$ Labels.}
We observe another useful property of single-fault replacement paths. For a vertex $a \in V$ with an $a$-$s$ replacement path $P = P_{a,s,\parent(a)}$, let $f_{a} \in P$ be the first (closest to $a$) vertex in $T[s, \parent(a))$.
\begin{observation}\label{obs:the_f_vertex}
	$P_{a,s,\parent(a)} = Q \circ T[f_a ,s]$ where $Q$ avoids $T[s,\parent(a)]$.
\end{observation}

We are now ready to define the $\FTLabel{AH}$ labels. These are constructed by Algorithm \ref{alg:AH-label}.
The label length of $O(\log^2 n)$ bits follows by Lemma \ref{lem:instersting_sets}.
\begin{algorithm}[h!]
	\caption{Construction of label $\FTLabel{AH} (a)$ for vertex $a$}\label{alg:AH-label}
	\For{each $b' \in I(a)$ with $\parent(b')=b$}{
		\textbf{store} vertices $b, b', f_{b'}$\;
		\textbf{store} $\conn(b', \parent(b), \Gsub{b,f_{b'}})$, $\ccid(b', G \setminus T[s,b])$\;
	}
\end{algorithm}

\paragraph{Decoding Algorithm for (AH) Case.}
Assume we are given an (AH)-query $\langle u,v,x,y \rangle$ along with the $\FTLabel{AH}$ labels of these vertices.
The main idea behind the construction of the $\FTLabel{AH}$ labels is to have:

\begin{claim}\label{claim_par(y)}
	If $w \in \{u,v\}$ is a descendant of $y$, then given the $\FTLabel{AH}$ labels of $w,x,y$, one can:
	
	(1) find $\ccid(w, G \setminus T[s,y])$, and
	
	(2) determine whether $w$ and $\parent(y)$ are $xy$-connected.
\end{claim}

\begin{proof}
	Let $y'$ be the child of $y$ on $T[y,w]$.
	Then $y',w$ are $xy$-connected as $T[y',w]$ is fault-free.
	Hence, in the following we can replace $w$ by $y'$ for determining both (1) and (2).
	Also, as $w,y'$ are (particularly) $y$-connected, and $w,s$ are $y$-connected by (C1), we have that $y',s$ are $y$-connected, so $f_{y'}$ is well-defined.
	By Lemma \ref{lem:child_on_path}, it holds that $y' \in I(w) \cup I(y)$, hence the triplet $b = y$, $b' = y'$ and $f_{b'} = f_{y'}$ is stored either in $\FTLabel{AH} (w)$ or in $\FTLabel{AH} (y)$.
	The same label also includes $\ccid(y', G \setminus T[s,y])$, so (1) follows.
	We also find there the value $\conn(y', \parent(y), \Gsub{y, f_{y'}})$.
	Note that if $f_{y'} = x$, then (2) follows as well.
	Assume now that $f_{y'} \neq x$.
	By Observation \ref{obs:the_f_vertex}, the path $P_{y',s,y}$ is of the form $Q \circ T[f_{y'},s]$ where $Q$ avoids $T[s,y]$.
	We now observe that $f_{y'} \notin T[s,x)$: this follows as otherwise, the $w$-$s$ path given by $T[w,y'] \circ Q \circ T[f_{y'},s]$ is failure-free, contradicting (C2).
	As $f_{y'}$ is, by definition, a vertex in $T[s,y)$, it follows that $f_{y'} \in T(x,y)$.
	The $y'$-$\parent(y)$ path $Q \circ T[f_{y'}, \parent(y)]$ now certifies that $y',\parent(y)$ are $xy$-connected, which gives (2). See illustration in Figure \ref{fig:dep_claim}.
\end{proof}
\begin{figure}
	\centering
	\includegraphics[scale=0.5]{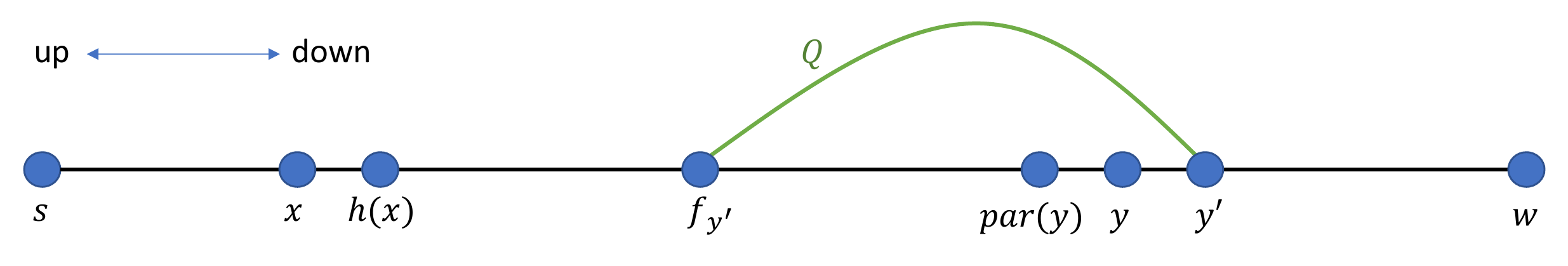}
	\caption{
		Illustration of the proof of Claim \ref{claim_par(y)}.
		The tree path from $s$ to $w$ is shown sideways, where the depth increases from left to right.
		The path $Q$ appears in green.
	}
	\label{fig:dep_claim}
\end{figure}

We next show how to use Claim \ref{claim_par(y)} for determining the $xy$-connectivity of $u,v$.
The proof divides into three cases according to the ancestry relations between $u$, $v$ and $y$.

\smallskip\noindent\textbf{Case 1: Neither of $u,v$ is a descendant of $y$.}
Then $y \notin T[u,v]$. Since both $u,v \in T_{\heavy(x)}$, also $x \notin T[u,v]$. Thus $u,v$ are $xy$-connected, and we are done.

\smallskip\noindent\textbf{Case 2: Only one of $u,v$ is a descendant of $y$.}
W.l.o.g., assume the descendant is $u$.
Then $y \neq \heavy(x)$, as otherwise $v$ would also be a descendant of $y$.
It follows that $\parent(y) \in T_{\heavy(x)}$.
Hence $T[\parent(y), v] \subseteq T_{\heavy(x)}$, so it avoids $x$.
It also avoids $y$, as $T[\parent(y),v]$ contains only ancestors of $\parent(y)$ or of $v$.
Thus, $v,\parent(y)$ are $xy$-connected.
Finally, we use Claim \ref{claim_par(y)}(2) with $w=u$ to determine the $xy$-connectivity of $u,\parent(y)$, or equivalently of $u,v$.

\smallskip\noindent\textbf{Case 3: Both $u,v$ are descendants of $y$.}
We apply Claim \ref{claim_par(y)} twice, with $w=u$ and $w=v$.
Using (2), we check for both $u$ and $v$ if they are $xy$-connected to $\parent(y)$.
The only situation in which we cannot infer the $xy$-connectivity of $u,v$ is when both answers are negative.
When this happens, we exploit (1) and compare the component IDs of $u,v$ in $G \setminus T[s,y]$.
If they are equal, then clearly $u,v$ are $xy$-connected as $x,y \in T[s,y]$.
Otherwise, we assert that we can safely determine that $u,v$ are $xy$-disconnected. 

\begin{claim}\label{claim:Tsy}
	If 
	(i) both $u$ and $v$ are $xy$-disconnected from $\parent(y)$, and (ii) $u$ and $v$ are $T[s,y]$-disconnected,
	then $u,v$ are $xy$-disconnected.
\end{claim}
\def\TSYPROOF{
\begin{proof}
	Assume towards a contradiction that there exists a $u$-$v$ path $P$  in $\Gsub{x, y}$.
	By (ii), $P$ must intersect $T[s,y]$.
	Let $a$ be a vertex in $P \cap T[s,y]$.
	As $a \notin \{x,y\}$, it holds that either $a \in T[s,x)$ or $a \in T(x,y)$.
	If $a \in T[s,x)$, then the path $P[u,a] \circ T[a,s]$ connects $u$ to $s$ in $\Gsub{x,y}$, but this contradicts (C2).
	If $a \in T(x,y)$, then the path $P[u,a] \circ T[a,\parent(y)]$ connects $u$ to $\parent(y)$ in $\Gsub{x,y}$, contradicting (i).
\end{proof}
}\TSYPROOF

This concludes the decoding algorithm for the (AH) case, and proves:
\begin{lemma}\label{lem:AH-2vft}
	There exist $O(\log^2 n)$-bit labels $\FTLabel{AH}$ supporting the all-heavy (AH) case.
\end{lemma}

\subsubsection{The General Dependent-Failures Case}\label{sec:reduction_to_AH}

We assume w.l.o.g. that $y$ is a descendant\footnote{We can check this using the extended IDs, and swap the roles of $x$ in $y$ if needed.} of $x$ in $T$.
Condition (C2) implies that both $u$ and $v$ are also descendants of $x$.
Recall that our general strategy is reducing to the (AH) case, as follows.
First, the case where $y$ is not in $T_{\heavy(x)}$ is handled by using 1-VFT labels \emph{in the graph $\Gsub{x}$}, which enable us to determine directly whether $u,v$ are connected in $(\Gsub{x}) \setminus \{y\} = \Gsub{x,y}$.
In the remaining case where $y \in T_{\heavy(x)}$, we show how to ``replace'' $u,v$ by \emph{$xy$-analogs}: vertices $\widehat{u}, \widehat{v}$ that are $xy$-connected to $u,v$ (respectively), and lie inside $T_{\heavy(x)}$.
Thus, the query $\langle \widehat{u}, \widehat{v}, x, y \rangle$ is an analogous (AH)-query to answer.

\paragraph{Construction of $\FTLabel{dep}$ Labels.}
The construction is based on defining, for any given vertex $a$ and ancestor $b$ of $a$, a small set of vertices from $T_{\heavy(b)}$, serving as candidates to be $bc$-analogs of $a$ for \emph{any} $c \in T_{\heavy(b)}$.
\begin{definition}[Analog Sets]
	For $a,b \in V$ such that $b$ is an ancestor of $a$ in $T$, the \emph{analog set} $\AnSet(a,b)$ consists of two (arbitrary) distinct vertices $c_1, c_2$ with the following property:
	$c_i \in T_{\heavy(b)}$ and there exists an $a$-$c_i$ path avoiding $T_{\heavy(b)}^+ \setminus \{c_i\}$.
	If there is only one such vertex, then $\AnSet(a,b)$ is the singleton containing it,
	and if there are none then $\AnSet(a,b) = \emptyset$.
\end{definition}

The $\FTLabel{dep}$ labels are constructed by Algorithm \ref{alg:dep-label}.
We use the notation $\FTLabel{1F}(a, G')$ to denote the $1$-VFT label of $a \in V$ from Theorem \ref{thm:1vft} \emph{constructed w.r.t the subgraph $G' \subseteq G$}.
The label length of $O(\log^3 n)$ bits follows by Lemma \ref{lem:instersting_sets}, Theorem \ref{thm:1vft} and Lemma \ref{lem:AH-2vft}.
\begin{algorithm}[h!]
	\caption{Construction of label $\FTLabel{dep} (a)$ for vertex $a$}\label{alg:dep-label}
	\textbf{store} $\FTLabel{AH} (a)$\;
	\For{each $b' \in I(a)$ with $\parent(b')=b$}{
		\textbf{store} vertices $b, b'$\;
		\textbf{store}  $\FTLabel{1F} (a, \Gsub{b})$, $\FTLabel{1F} (b', \Gsub{b})$\;
		\textbf{store} vertex set $\AnSet(a,b)$, and $\FTLabel{AH} (c_i)$ for each $c_i \in \AnSet(a,b)$\;
		\textbf{store} $\ccid(a, G \setminus T_{\heavy(b)}^+)$\;
	}
\end{algorithm}

\paragraph{Decoding Algorithm for Dependent Failures.}
Assume we are given a dependent-failures query $\langle u,v,x,y \rangle$ where $y$ is a descendant of $x$, along with the $\FTLabel{dep}$ labels.
We first treat the easier case where $x$ is a light ancestor of $y$ using the $1$-VFT labels in $\Gsub{x}$.

\smallskip\noindent\textbf{Case: $x$ is a light ancestor of $y$.}
Then the child $x_y$ of $x$ on the path $T[x,y]$ is light, hence $x_y \in I(y)$.
Therefore, the label of $y$ contains the $1$-VFT label $\FTLabel{1F} (y, \Gsub{x})$.
Let $x_u$ be the child of $x$ on the path $T[x,u]$, and define
\[
\widetilde{u} =
\begin{cases}
	\text{if $x_u$ is light:} & u, \\
	\text{if $x_u$ is heavy:} & x_u = \heavy(x).
\end{cases}
\]
\begin{claim}\label{claim:u_tilde}
	It holds that (i) $\widetilde{u}$ is $xy$-connected to $u$, and (ii) one can find $\FTLabel{1F} (\widetilde{u}, \Gsub{x})$.
\end{claim}
\def\UTILDEPROOF{
\begin{proof}
If $x_u$ is light: Then (i) is trivial.
For (ii), note that $x_u \in I(u)$, hence the $1$-VFT label of $\widetilde{u} = u$ with respect to $\Gsub{x}$, which is $\FTLabel{1F} (u, \Gsub{x})$, is stored in $\FTLabel{dep} (u)$.

If $x_u$ is heavy: Then $x_u \neq x_y$, hence the path $T[x_u, u]$ is failure-free, which proves (i).
For (ii), note that $x_u = \heavy(x) \in I(x)$, hence the $1$-VFT label of $\widetilde{u} = \heavy(x)$ with respect to $\Gsub{x}$, which is $\FTLabel{1F} (\heavy(x), \Gsub{x})$, is stored in $\FTLabel{dep} (x)$.
\end{proof}
}\UTILDEPROOF

We define $\widetilde{v}$ and find its $1$-VFT label with respect to $\Gsub{x}$ in a similar fashion.
Finally, we use the $1$-VFT labels to answer the \emph{single failure} query $\langle \widetilde{u}, \widetilde{v}, y \rangle$ with respect to the graph $\Gsub{x}$, which determines the $xy$-connectivity of $\widetilde{u},\widetilde{v}$, or equivalently of $u,v$. 
So in this case, the decoding algorithm directly determine the $xy$-connectivity of $u,v$, and we are done.

From now on, we assume that $x$ is a heavy ancestor of $y$, or equivalently that $y \in T_{\heavy(x)}$.

\smallskip\noindent\textbf{Replacing $u,v$ with their $xy$-analogs.}
In the following, we restrict our attention to $u$ (and the same can be applied to $v$).
Again, let $x_u$ be the child of $x$ on $T[s,u]$.
The $xy$-analog $\widehat{u}$ of $u$ is defined as:
\[
\widehat{u} =
\begin{cases}
	\text{if $x_u$ is heavy:} & u, \\
	\text{if $x_u$ is light and $\AnSet(u,x) \setminus \{y\} \neq \emptyset$:} & c \in \AnSet(u,x) \setminus \{y\},  \\
		\text{if $x_u$ is light and $\AnSet(u,x) \setminus \{y\} = \emptyset$:} & \text{undefined.}
\end{cases}
\]
The corner case where $\widehat{u}$ is undefined can be alternatively described as follows:
\begin{enumerate}[(C4)]
	\item $x_u$ is light, and no $c \in T_{\heavy(x)} \setminus \{y\}$ is connected to $u$ by a path internally avoiding $T_{\heavy(x)}^+$.
\end{enumerate}
The key observation for handling case (C4) is:
\begin{claim}\label{claim:avoid_Thx}
	If (C4) holds, then:
	
	(1) For any vertex $c \in T_{\heavy(x)} \setminus \{y\}$, $u$ and $c$ are $xy$-disconnected.
	
	(2) For any vertex $c \notin T_{\heavy(x)} \cup \{x, y\}$, 
	$u$ and $c$ are $xy$-connected iff they are $T_{\heavy(x)}^+$-connected. 
\end{claim}
\def\CFOURPROOF{
\begin{proof}
	For (1), assume towards a contradiction that there exists a $u$-$c$ path $P$ in $\Gsub{x,y}$.
	Let $c'$ be the first (closest to $u$) vertex from $T_{\heavy(x)}$ appearing in $P$.
	Since $c' \neq y$, and the subpath $P[u,c']$ internally avoids $T_{\heavy(x)}^+$, we get a contradiction to (C4).
	
	The ``if'' direction of (2) is trivial as $\{x,y\} \subseteq T_{\heavy(x)}^+$.
	For the ``only if'' direction, let $P$ be a $u$-$c$ path in $\Gsub{x,y}$.
	It suffices to prove that $P$ avoids $T_{\heavy(x)}$, but this follows directly from (1).
\end{proof}
}\CFOURPROOF

We handle case (C4) as follows.
If $v \in T_{\heavy(x)}$, then by Claim \ref{claim:avoid_Thx}(1) we determine that $u,v$ are $xy$-disconnected, and we are done.
Else, the child $x_v$ of $x$ on $T[x,v]$ is light, hence $x_v \in I(v)$. Therefore $\FTLabel{dep} (v)$ contains $\ccid(v, G\setminus T_{\heavy(x)}^+)$ (set $a=v$, $b = x$ and $b' = x_v$).
As $x_u$ is also light by (C4), we can find  $\ccid(u, G\setminus T_{\heavy(x)}^+)$ in a similar fashion. By Claim \ref{claim:avoid_Thx}(2), comparing these $\ccid$s allows us to determine the $xy$-connectivity of $u,v$, and we are done again.

If the corner case (C4) does not hold, then $\widehat{u}$ is indeed a valid $xy$-analog of $u$. Namely:
\begin{claim}\label{claim:u_hat}
	If $\widehat{u}$ is defined, then (i) $\widehat{u} \in T_{\heavy(x)}$, (ii) $u,\widehat{u}$ are $xy$-connected, and (iii) one can find the label $\FTLabel{AH} (\widehat{u})$.
\end{claim}
\def\UHATPROOF{
\begin{proof}
	If $x_u$ is heavy: Then (i) and (ii) are trivial. For (iii) we simply note that $\FTLabel{dep}(u)$ stores $\FTLabel{AH}(u)$.
	
	If $x_u$ is light and $\AnSet(u,x) \setminus \{y\} \neq \emptyset$:
	Then $\widehat{u} \in \AnSet(u,x) \setminus \{y\}$.
	By definition of $\AnSet(u,x)$ it holds that $\widehat{u} \in T_{\heavy(x)}$, which gives (i), and that there is a $u$-$\widehat{u}$ path avoiding $T_{\heavy(x)}^+ \setminus \{\widehat{u}\}$, and consequently also $\{x,y\}$, which gives (ii).
	For (iii), we note that as $x_u$ is light it holds that $x_u \in I(u)$. Thus, $\FTLabel{dep} (u)$ stores the $\FTLabel{AH}$ labels of the vertices in $\AnSet(u,x)$, and in particular stores $\FTLabel{AH}(\widehat{u})$.
\end{proof}
}\UHATPROOF

\smallskip\noindent\textbf{Finalizing.}
We have shown a procedure that either determines directly the $xy$-connectivity of $u,v$, or certifies that $y \in T_{\heavy(x)}$ and finds $xy$-analogs $\widehat{u}, \widehat{v} \in T_{\heavy(x)}$ of $u,v$ (respectively) along with their (AH)-labels $\FTLabel{AH} (\widehat{u}), \FTLabel{AH} (\widehat{v})$.
In the latter case, we answer the (AH)-query $\langle \widehat{u}, \widehat{v}, x, y \rangle$ using the $\FTLabel{AH}$ labels\footnote{$\FTLabel{AH} (x)$ and $\FTLabel{AH} (y)$ are stored in $\FTLabel{dep} (x)$ and $\FTLabel{dep} (y)$ respectively.} and determine the $xy$-connectivity of $\widehat{u}, \widehat{v}$, or equivalently of $u,v$.
This concludes the decoding algorithm for dependent failures. We therefore have:

\begin{lemma}\label{lem:dep-2vft}
	There exists $O(\log^3 n)$-bit labels $\FTLabel{dep}$ supporting the dependent failures case.
\end{lemma}

By combining the $\FTLabel{dep}$ labels of Lemma \ref{lem:dep-2vft} with the  $\FTLabel{ind}$ labels of Lemma \ref{lem:ind-2vft}, we obtain the $2$-VFT labels of Theorem \ref{thm:2vft}. The preprocessing time analysis, which completes the proof of Theorem \ref{thm:2vft}, is found in Appendix \ref{sec:missing-proofs}.

\def\APPENDTIMETWOF{
\paragraph{Preprocessing Time for $2$-VFT Labels (Theorem \ref{thm:2vft}).}
We start by computing a $3$-vertex connectivity certificates $G' \subseteq G$ such that $|E(G')|=O(n)$, which can be done in $O(m)$ time by \cite{cheriyan1993scan}.
We then apply the labeling algorithm on the graph $G'$.
The first step is to compute the single-source $2$-VFT labels of \cite{Choudhary16}, which can be done in $O(|V(G')| \cdot |E(G')|) = O(n^2)$ time overall \cite{ChoudharyPersonal}.
Next, it is easy to see that each label $\FTLabel{ind} (v)$ or $\FTLabel{dep} (v)$ can be computed by applying $\widetilde{O}(1)$ connectivity algorithms and specific $1$-VFT label constructions, each taking $\widetilde{O}(n)$ time.
This completes the proof of Theorem \ref{thm:2vft}.
}

\section{Sublinear $f$-VFT Labels}\label{sec:sublinear}

In this section, we provide an $f$-VFT labeling scheme with sublinear size for any $f=o(\log\log n)$. Note that labels of near-linear size are directly obtained by the $f$-EFT labeling scheme of \cite{DoryP21} (e.g., by including in the labels of a vertex, the EFT-labels of all its incident edges). We show:

\begin{theorem}[$f$-VFT Labels with Sublinear Size]\label{thm:sublinear-labels}
For every $n$-vertex graph $G=(V,E)$ and fixed parameter $f=o(\log\log n)$, there is a polynomial time randomized algorithm for computing $f$-VFT labels of size $\widetilde{O}(n^{1-1/2^{f-2}})$. For every query $\langle u,v, F \rangle$ for $F \subseteq V$, $|F|\leq f$, the correctness holds w.h.p.
\end{theorem} 

We use the EFT-labeling scheme of Dory and Parter \cite{DoryP21}, whose label size is independent in the number of faults. The correctness guarantee holds w.h.p. for a \emph{polynomial} number of queries. 
\begin{theorem}[Slight Restatement of Theorem 3.7 in \cite{DoryP21}]\label{thm:EFT-labels}
For every undirected $n$-vertex graph $G=(V,E)$, there is a randomized EFT connectivity labeling scheme with labels $\EFTLabel: V \cup E \to \{0,1\}^{\ell}$ of length $\ell=O(\log^3 n)$ bits (independent in the number of faults).
For a given triplet along w the $\EFTLabel$ labels of $u,v \in V$ and every edge set $F \subseteq E$, the decoding algorithm determines, w.h.p., if $u$ and $v$ are connected in $G \setminus F$. 
\end{theorem}

\paragraph{The Labels.} Our starting observation is that one can assume, w.l.o.g., that $|E(G)|\leq fn$ edges. This holds as it is always sufficient to apply the labeling scheme on the sparse $f$ (vertex) connectivity certificate of $G$, which has at most $f n$ edges, see e.g., \cite{cheriyan1993scan}. The labeling scheme is inductive where the construction of $f$-VFT labels is based on the construction of $(f-1)$ VFT labels given by the induction assumption. For the base of the induction ($f=2$), we use the $2$-VFT labels of Theorem \ref{thm:2vft}. The approach is then based on dividing the vertices into high-degree and low-degree vertices based on a degree threshold $\Delta=2f \cdot n^{1-1/2^{f-2}}$. Formally, let $V_H$ be all vertices with degree at least $\Delta$. 
By our assumption, the number of high-degree vertices is at most $|V_H|\leq O(f n/ \Delta)$. Letting $\EFTLabel(\cdot)$ denote that $f$-EFT labeling scheme of Theorem \ref{thm:EFT-labels} by \cite{DoryP21}, the $f$-VFT label of $v$ is given by Algorithm \ref{alg:f-vft-label}.
\begin{algorithm}[!h]
	\caption{Construction of label $\VFTLabel{f} (v)$ of for vertex $v$}\label{alg:f-vft-label}
	\textbf{store} $\EFTLabel (v)$\;
	\For{each $x \in V_H$}{
		\textbf{store} $\VFTLabel{f-1} (v, \Gsub{x})$\;
	}
	\If{$v \in V \setminus V_H$}{
		\textbf{store} $\EFTLabel(e = (u,v))$ for every adjecent edge $(u,v) \in G$\;
	}
\end{algorithm}

\paragraph{The Decoding Algorithm.} Consider a query $\langle u,v, F \rangle \in V \times V \times V^{\leq f}$. We distinguish between two cases, based on the degrees of the faults $F$ in the graph $G$. Assume first that there exists at least one high-degree vertex $x \in F \cap V_H$. In this case, the labels of every $w \in \{u,v\} \cup (F \setminus \{x\})$ includes the $(f-1)$ VFT label in $G \setminus \{x\}$, namely, $\VFTLabel{f-1}(v,G \setminus \{w\})$. We can then determine the $F$-connectivity of $u$,$v$ using the decoding algorithm of the $(f-1)$-VFT labels (given by the induction assumption). It remains to consider the case where all vertices have low-degrees. In this case, the VFT-labels include the EFT-labels of $u,v$, and all failed edges, incident to the failed vertices. This holds as the label of every failed vertex $x \in F$ contains $\EFTLabel(e=(x,z))$ for each of its incident edges $(x,z)$ in $G$. This allows us to apply the decoding algorithm of Theorem \ref{thm:sublinear-labels} in a black-box manner. 

\paragraph{Label Size.}
We now turn to bound the label size. For every $f \leq n$, let $\sigma_V(f, n), \sigma_E(n)$ be an upper bound on $f$-VFT (resp., EFT) labels for $n$-vertex graphs. Assume by induction on $g \leq f-1$ that 
\begin{equation}\label{eq:sublinear-recurse}
\sigma_V(g, n)=2^{g-2}\cdot g \cdot n^{1-1/2^{g-2}}\cdot c\cdot\log^3 n~,
\end{equation}
where $c\cdot \log^3 n$ is the bound on the EFT labels of Theorem \ref{thm:EFT-labels}.
This clearly holds for $g=2$ (by Theorem \ref{thm:2vft}). Denote the length of the $f$-VFT label for vertex $v$ by $|\VFTLabel{f} (v)|$. By taking $\Delta(f,n)=2f \cdot n^{1-1/2^{f-2}}$ to be the degree threshold $\Delta$ in our $f$-VFT label construction, and using Eq. (\ref{eq:sublinear-recurse}), we have:
\[
|\VFTLabel{f} (v)| \leq \sigma_V(f-1, n-1) \cdot (2nf/\Delta(f,n)) + \Delta(f,n)\cdot \sigma_E(f, n)~ \leq \sigma_V(f, n)~.
\]
This satisfies the induction step and provides a bound of $\sigma_V(f, n)=\widetilde{O}( n^{1-1/2^{f-2}})$ for every $f =o(\log\log n)$, as desired. Theorem \ref{thm:sublinear-labels} follows.

\paragraph{Acknowledgments.}
We would like to thank Michal Dory for useful discussions.

\bibliographystyle{alpha}
\bibliography{dist-cut}

\newpage

\appendix
\section{Preprocessing Times}\label{sec:missing-proofs}

\APPENDTIMEONEF


\APPENDTIMETWOF
\section{Distributed Computation of Labels and Cut Vertices}\label{sec:distributed-comp}

We now turn to describe distributed computation of our $1$-VFT labels in the congest model of distributed computing \cite{Peleg:2000}. In this model, the communication is abstracted as an $n$-vertex graph $G=(V,E)$. Each vertex holds a unique identifier of $O(\log n)$ bits. The algorithm works in synchronous rounds, where in each round, a vertex can exchange $O(\log n)$-bit messages with each of its neighbors. In \cite{DistCut2022}, the authors describe a distributed algorithm for detecting single vertex cuts. The exact same algorithm can also provide $1$-VFT labels. For the sake of completeness, we describe the algorithm here. 

\subsection{Graph Sketches}
The algorithm is based on the notion of \emph{graph sketches}. The implementation of the graph sketches technique  calls for providing each of the graph edges with \emph{convenient} identifiers that ease the subsequent connectivity algorithms (in which these sketches get used). We start by defining these identifiers and then provide the formal sketch definition. 

\paragraph{Extended Edge Identifiers \cite{GhaffariP16,DoryP21}.} The decoding of the sketch information requires one to distinguish between an identifier of a single edge to the bitwise XOR of several edges. For this purpose, we define for each edge $e$ an extended edge identifier $\EID_T(e)$ that allows distinguishing between these cases, and serves as the identifier of the edge. The extended edge identifier $\EID_T(e)$ consists of a (randomized) unique distinguishing identifier $\UID(e)$, as well as additional tree related information. The computation of $\UID(e)$ is based on the notion of $\epsilon$-\emph{bias} sets \cite{naor1993small}. The construction is randomized and guarantees that, w.h.p., the XOR of the $\UID$ part of each given subset of edges $S \subseteq E$, for $|S|\geq 2$, is not a legal $\UID$ identifier of any edge. Let $\XOR(S)$ be the bitwise XOR of the extended identifiers of edges in $S$, i.e., $\XOR(S)=\oplus_{e \in S} \EID_T(e)$. In addition, let $\XOR_U(S)=\oplus_{e \in S} \UID(e)$. 
\begin{lemma}[Modification of Lemma 2.4 in \cite{GhaffariP16}]
\label{cl:epsbias}
There is an $\widetilde{O}(D)$-round algorithm that sends all nodes  a random seed $\mathcal{S}_{ID}$ of $O(\log^2 n)$ bits. Using this seed, all nodes locally define a (consistent) collection $\mathcal{I}=\{\UID(e_1), \ldots, \UID(e_{M})\}$ of $M=\binom{n}{2}$ random identifiers for all possible edges $(u,v)$, each of $O(\log n)$-bits. These identifiers are such that for each subset $E' \subseteq E$, where $|E'|\neq 1$, we have $\Pr[\XOR_U(E') \in \mathcal{I}] \leq 1/n^{10}$. In addition, given the identifiers $\ID(u), \ID(v)$ of the edge $e=(u,v)$ endpoints, and the seed $\mathcal{S}_{ID}$, one can determine $\UID(e)$ in $\widetilde{O}(1)$ time.
\end{lemma}
The extended identifier $\EID_T(e)$ is given by
\begin{equation}\label{eq:extend-ID}
\EID_T(e)=[\UID(e), \ID(u), \ID(v), \LCALabel_T(u), \LCALabel_T(v)]~.
\end{equation}
The identifiers of $\ID(u), \ID(v)$ are used in order to verify the validity of the unique identifier $\UID(e)$.  When the tree $T$ is clear from the context, we might omit it and simply write $\EID(e)$. 

\paragraph{Defining Graph Sketches.} 
Graph sketches are a tool to identify outgoing edges. The sketches are based on random sub-sampling of the graph edges with logarithmic number of scales, i.e., with probability of $2^{-j}$ for every $j \in [0,\log m]$.  We follow \cite{DuanConnectivityArxiv16,DuanConnectivitySODA17} and use pairwise independent hash functions to decide whether to include edges in sampled sets.

Choose $L=c\log n$ pairwise independent hash functions $h_1, \ldots, h_{L}:\{0,1\}^{\Theta(\log n)} \to \{0, \ldots, 2^{\log m}-1\}$, and for each $i \in \{1, \ldots, L\}$ and $j \in [0,\log m]$, define the edge set 
$E_{i,j} =\{ e \in E ~\mid h_i(e) \in [0,2^{\log m-j})\}~.$
Each of these hash functions $h_i$ can be defined using a random seed $\mathcal{S}_{h_i}$ of logarithmic length \cite{TCS-010}. Thus, a 
random seed $\mathcal{S}_{\mathcal{H}}$ of length $O(L \log n)$ can be used to determine the collection of all these $L$ functions. 

As observed in \cite{DuanConnectivityArxiv16,GibbKKT15}, pairwise independence is sufficient to guarantee that for any non-empty set $E' \subset E$ and any $i$, there exists an index $j$, such that with constant probability $\XOR(E' \cap E_{i,j})$ is the name (extended identifier) of one edge in $E'$, for a proof see Lemma 5.2 in  \cite{GibbKKT15}.
Hence, by exploiting the structure of the extended IDs, it holds:
\begin{lemma}\label{lemma_unique}
(i) [Lemma 5.2 in  \cite{GibbKKT15}] For any non-empty edge set $E'$ and any $i$, with constant probability there exists a $j$ satisfying that $|E' \cap E_{i,j}|=1$; (ii) Given the seed $\mathcal{S}_{ID}$, one can determine in $\widetilde{O}(1)$ time if $\XOR(E' \cap E_{i,j})$ corresponds to a single edge ID in $G$ or not, w.h.p.
\end{lemma}

For each vertex $v$ and indices $i,j$, let $E_{i,j}(v)$ be the edges incident to $v$ in $E_{i,j}$. 
The $i^{th}$ \emph{basic sketch unit} of each vertex $v$ is then given by:
\begin{equation}
\label{eq:vsketch}
\Sketch_{G,i}(v)=[\XOR(E_{i,0}(v)),\ldots,\XOR(E_{i,\log m}(v))].
\end{equation}
The sketch of each vertex $v$ is defined by a concatenation of $L=\Theta(\log n)$ basic sketch units: 
\[
\Sketch_G(v)=[\Sketch_{G,1}(v),\Sketch_{G,2}(v), \ldots\Sketch_{G,L}(v)]~.
\]
For every subset of vertices $S$, let 
$\Sketch_G(S)=\oplus_{v \in S}\Sketch_G(v).$ When the graph $G$ is clear from the context, we may omit it and write $\Sketch_{i}(v)$ and $\Sketch(v)$. 

\begin{lemma}\label{lem:sketch-property}
For any subset $S$, given one basic sketch unit $\Sketch_i(S)$ and the seed $\mathcal{S}_{ID}$ one can compute, with constant probability, an outgoing edge $E(S, V \setminus S)$ if such exists. The complexity is $\widetilde{O}(1)$ time.
\end{lemma}

\subsection{The Distributed Algorithms}
We are now ready to describe the distributed label construction. Throughout, let $T$ be a BFS tree rooted at some arbitrary vertex $s$. The latter can be computed in $O(D)$ rounds. 
By Lemma 2.1 in \cite{Censor-HillelD17}, we have:
\begin{lemma}\label{lem:dist-HL-anc}
Given a tree $T$, there is an $O(D(T))$-round algorithm for computing ancestry labels $\LCALabel_{T}(v)$ of $O(\log^2 n)$ bits for every $v \in V$, where $D(T)$ is the depth of the tree $T$.
\end{lemma}

\paragraph{Step (0): Computation of Heavy-Light Tree Decomposition and Extended IDs.} Given the BFS tree $T$, the heavy-light tree decomposition can be computed in $O(D)$ rounds. In the distributed output format, each vertex $v$ knows its heavy child in $T$. By propgating tree edges downwards the tree, each vertex $v$ can learn its tree path $T[s,v]$ as well as to identify the light edges (and consequently also the heavy edges) on this path. 

It is easy to see that this can be done simply by computing the size of subtrees $|T_v|$ for each vertex $v$. This allows $v$ to determine its heavy child in $T$. Using Lemma \ref{lem:dist-HL-anc}, each vertex can obtain its ancestry label $\LCALabel_{T}(v)$. To complete the computation of the extended edge-IDs, the source $s$ samples a random seed $\mathcal{S}_{ID}$ of $\widetilde{O}(1)$ bits and share it with all vertices. Then, using \Cref{cl:epsbias}, each vertex $v$ can then locally compute the unique edge-ID $\UID(e)$ for each of its incident edges. 

\paragraph{Step (1): Computation of Subtree Sketches.}
The source $s$ locally samples the collection of $L=O(\log n)$ random seeds $\mathcal{S}_{\mathcal{H}} = \{\mathcal{S}_{h} \mid h \in \{1,\ldots, L\}\}$, and send it to all the vertices. These seeds provide all the required randomness for the computation of the sketch information. By aggregating the individual sketch information $\alpha_u=\Sketch_G(u)$ for every $u \in V$, from the leaf vertices on $T$ up to the root $s$, each vertex $v$ obtains its subtree sketch, given by $\Sketch_G(V(T_v))=\oplus_{u \in T_v} \alpha_u$. By the end of this sketch aggregation process, the source $s$ broadcasts its output sketch, namely $\Sketch_G(V)$, to all the vertices in the graph.

\paragraph{Step (2): Local Connectivity Computation.} This step is locally applied at every vertex $x$, and requires no additional communication. We show that each vertex $x$, given the received sketch information in Step (1), can locally simulate the Bor\r{u}vka's algorithm \cite{Boruvka} in the graph $G \setminus \{x\}$. Consequently, it can determine if $G \setminus \{x\}$ is connected, as well as assigning each of its children a unique ID of their connected component in $G \setminus \{x\}$. Let $u_1, \ldots, u_{d_x}$ be the children of $x$ in $T$. The connected components in $T \setminus \{x\}$ are denoted by 
\[
\mathcal{C}_x=\{ V(T_{u_j}), j \in \{1,\ldots, d_x\}\} \cup \{V \setminus V(T_x)\}.
\]
Observe that while $x$ does not know all vertices in the components of $\mathcal{C}_x$, it knows (at least) one vertex in each of these components (i.e., its immediate children and $s$). This allows $x$ to assign each component in $\mathcal{C}_x$ an $O(\log n)$-bit identifier. The ID of the component containing $s$ will be throughout defined by $\ID(s)$. 

By Step (1), $x$ holds the $G$-sketch of each component in $\mathcal{C}_x$: it has explicitly received $\Sketch_G(V(T_{u_j}))$ from each child $u_j$ for every $j \in \{1,\ldots, d_x\}$. In addition, it can locally infer $\Sketch(V \setminus V(T_x))=\Sketch(V) \oplus \Sketch(V(T_x))$. To locally implement the Bor\r{u}vka's algorithm on these connected components, it is first required to update these $G$-sketches into $(G \setminus \{x\})$-sketch, as described next. 
\\
\noindent \textbf{(2.1): Obtaining sketch information in $G \setminus \{x\}$.} Recall that the seeds $\mathcal{S}_h$ determine the sampling of edges into the sketches. Therefore, as $x$ knows $\mathcal{S}_h$ for every $h \in \{1,\ldots, L\}$, as well as, the extended identifiers of its incident edges (from Step (0)), it can cancel these edges from the respective entries in the sketches of each component $C \in \mathcal{C}_x$. 
Specifically, to cancel an adjacent edges $(x,u) \in G$, it first uses the ancestry labels of the edge to determine the connected component $C$ of $u$ in $\mathcal{C}_x$.
Then, using the $\mathcal{S}_h$ seeds, it can determine all the sketch entries of $\Sketch_G(C)$ in which $(x,u)$ has been sampled.
This allows $x$ to obtain $\Sketch_{G\setminus \{x\}}(C)$ for every $C \in \mathcal{C}_x$.
\\
\noindent \textbf{(2.2): Simulating the Bor\r{u}vka algorithm in $G \setminus \{x\}$.}
The input to this step is the identifiers of the components $\mathcal{C}_{x,0}=\mathcal{C}_{x}$ in $T \setminus \{x\}$, along with their sketch information in $G \setminus \{x\}$. The desired output is to determine a unique component-ID to each child of $x$ in $G \setminus \{x\}$. The algorithm consists of $L=O(\log n)$ phases of the Bor\r{u}vka algorithm, and works in an almost identical manner to the (centralized) decoding algorithm of \cite{DoryP21}. For completeness, we fully describe it here as well.

Each phase $i \in \{1,\ldots, L\}$ will be given as input a partitioning $\mathcal{C}_{x,i}=\{C_{i,1}, \ldots, C_{i,k_i}\}$ of (not necessarily maximal) connected components in $G \setminus \{x\}$; along with the sketch information of the components $\mathcal{C}_{x,i}$ in $G \setminus \{x\}$. The output of the phase is a partitioning $\mathcal{C}_{x,i+1}$, along with their sketch information in $G \setminus \{x\}$ and the identifiers of the components for each child of $x$.

A component $C_{i,j} \in \mathcal{C}_{i,x}$ is denoted as \emph{growable} if it has at least one outgoing edge to a vertex in $V \setminus (C_{i,j} \cup \{x\})$. Letting $N_i$ denote the number of growable components in $\mathcal{C}_i$, the output partitioning $\mathcal{C}_{i+1}$ of the $i^{th}$ step guarantees that $N_{i+1}\leq c N_i /2$ in expectation for some constant $0 < c \leq 1$ (see Lemma \ref{lem:sketch-property}). To obtain outgoings edges from the growable components in $\mathcal{C}_{x,i}$, the algorithm uses the $i^{th}$ basic-unit sketch $\Sketch_i(C_{i,j})$ of each $C_{i,j} \in \mathcal{C}_{x,i}$. By Lemma \ref{lem:sketch-property}, from every growable component in $\mathcal{C}_{x,i}$, we get one outgoing edge $e'=(u,v)$ with constant probability. Using the extended edge identifier of $e'$ the algorithm can also detect the component $C_{i,j'}$ to which the second endpoint, say $v$, of $e'$ belongs using the ancestry label of the detected edge $e'$. 
That allows us to compute the component of $v$ in the initial partitioning $T \setminus \{x\}$, i.e., the component $C_{q,0}$ of $v$ in $\mathcal{C}_{x,0}$. Thus $y$ belongs to the unique component $C_{i,j'} \in \mathcal{C}_i$ that contains $C_{q,0}$. Note that it is important to use fresh randomness (i.e., independent sketch information) in each of the Boruvka phases \cite{ahn2012analyzing,kapron2013dynamic,DuanConnectivityArxiv16}. The algorithm then computes the updated sketches of the merged components. For this purpose, each Bor\r{u}vka phase $i$ is based on the sketch information obtained with a fresh seed $\mathcal{S}_h$. The sketch information for phase $i+1$ is given by XORing over the sketches of the components in $\mathcal{C}_{x,i}$ that got merged into a single component in $\mathcal{C}_{x,i+1}$. In expectation, the number of growable components is reduced by a constant factor in each phase. Thus after $L=O(\log n)$ phases, the expected number of growable components is at most $1/n^5$, and using Markov inequality, we conclude that w.h.p there are no growable components. The final partitioning $\mathcal{C}_{x,L}$ corresponds, w.h.p, to the maximal connected components in $G \setminus \{x\}$. 

This local simulation of the Bor\r{u}vka's algorithm allows $x$ to deduce the following useful information. 
First, it can determine if it is a cut-vertex which holds, w.h.p., iff $|\mathcal{C}_{x,L}|\geq 2$. In addition, $x$ can compute a unique component ID for each of its children in $G \setminus \{x\}$. The ID of the component containing $s$ is set to $ID(s)$, and the ID of any other component in $\mathcal{C}_{x,L}$ (if exists) is set to be the largest ID among all $x$'s children contained in that component. 
Within another round of communication, each vertex $v$ holds its component-ID in $G \setminus \{\parent(v)\}$. We denote this value by $\ccid(v)$, namely $\ccid(v) := \ccid(v, \Gsub{\parent(v)})$. 

\paragraph{Step (3): Propagation of Component IDs Information.} To compute the $1$-VFT labels, it is required for each vertex $v$ to learn the component-ID $\ccid(u)$ for every light edge $(\parent(u),u) \in T[s,v]$. This information is obtained by letting every light vertex $u$ propagate its $\ccid(u)$ over $T_u$. Since each vertex $v$ has $O(\log n)$ light edges on its tree path $T[s,v]$, overall, it is required to receive $O(\log^2 n)$ bits of information. This can be done in $\widetilde{O}(D)$ rounds, by a simple pipeline. 

\paragraph{Step (4): Local Computation of $1$-VFT Labels.} At this point, each vertex $v$ is equipped with all the required information to locally compute its $1$-VFT label, $\FTLabel{1F} (v)$, as specified in Sec. \ref{sec:single-fault}. 
By Step (0), $v$ knows its interesting set $I(v)$ (see Def. \ref{def:intersets}). From Step (2) it holds the component-ID of its heavy child $\ccid(\heavy(v))$; and from Step (3) it holds the component-ID $\ccid(u)$ for every light-edge $(\parent(u),u) \in T[s,v]$. Since we set the IDs of components containing $s$ to be $ID(s)$, the received component-IDs also indicate on the connectivity to $s$ in the corresponding $G \setminus \{x\}$ subgraphs. This completes the description of the algorithm. 

We are now ready to complete the proof of Theorem \ref{thm:1vft} by analyzing the round complexity of the algorithm, and discussing its centralized implementation.

The round complexity of Steps (0-1) is $\widetilde{O}(D)$ by Lemma \ref{lem:dist-HL-anc} and the aggregation of $\widetilde{O}(1)$-bit values. Step (3) can be implemented in $\widetilde{O}(D)$ by a simple pipeline mechanism. As each vertex $v$ collects the component-ID of $u$ in $G \setminus \{\parent(u)\}$ for every light edge $(\parent(u),u)$ in the path $T[s,v]$, it collects $\widetilde{O}(1)$ bits of information, which can be done in $\widetilde{O}(D)$ rounds.  

Our congest algorithm can be simulated in $\widetilde{O}(m)$ time. The computation of the $n$ sketches $\Sketch_{G}(v)$ for every vertex $v$ takes $\widetilde{O}(m)$ time. The aggregate of sketch information can be then implemented in $\widetilde{O}(n)$ time. Finally, we claim that the local computation each vertex $x$ can be done in time $\widetilde{O}(\deg(x))$ where $\deg(x)$ is the degree of $x$ in $G$. This holds, as in each phase of Bor\r{u}vka's we have $O(\deg(x))$ components, and the computation of each outgoing edge using the sketch information  takes $\widetilde{O}(1)$ time, per component. Overall, the collection of all $n$ connectivity algorithms in $G \setminus \{x\}$ for every $x \in V$ can be implemented in $\widetilde{O}(m)$ time. This completes the randomized construction time of Theorem \ref{thm:1vft}.

\section{Single-Source Dual Failure Connectivity Labels}\label{sec:single-source-dual}
In this section we provide an independent approach for single-source $2$-VFT connectivity labels, and prove the following variant of Lemma \ref{lem:ss-2vft_labels}:

\begin{lemma}\label{lem:orig_ss-2vft_labels}
	There is a \emph{single-source} 2-VFT connectivity labeling scheme with label length $O(\log^2 n)$ bits.
	That is, given a query of vertices $\langle t,x,y \rangle$ along with their labels, one can determine whether the (fixed) source $s$ and $t$ are $xy$-connected. 
\end{lemma}

\paragraph{Heavy Paths.}
For this section, will need another concept related to the heavy light decomposition.
Consider the subgraph $T'$ of $T$ formed by taking in all the vertices but only the heavy edges.
The connected components of $T'$ are paths (some of length zero, i.e. isolated vertices), each corresponding to choosing a rooting light vertex in $T$ and going downwards along the heavy edges until reaching a leaf.
Such a path $Q$ is called a \emph{heavy path} of $T$.
The collection of all heavy paths is denoted by $\mathcal{Q}$.
We give each heavy path $Q \in \mathcal{Q}$ a unique identifier $\ID(Q)$ of $O(\log n)$ bits (e.g., the ID of its rooting light vertex).
We write $Q_u$ to denote the unique heavy path containing vertex $u \in V$.
Occasionally, it will be useful to consider the extension of a heavy path $Q$ up until the root $s$, denoted $Q^{\uparrow}$.
That is, $Q^{\uparrow}$ is the unique root-to-leaf path in $T$ containing $Q$.
Similarly to Observation \ref{obs:heavy-light}, we have:
\begin{observation}\label{obs:heavy_paths}
	Any root-to-leaf path in $T$ intersects only $O(\log n)$ heavy paths in $\mathcal{Q}$.
\end{observation}

\paragraph{Extended IDs Revisited.}
We slightly augment our extended vertex IDs presented in Section \ref{sec:prelim}, by including in $\EID(u)$ the number of light vertices in $T[s,u]$, denoted $\nlights(u)$, and the ID of the heavy path containing $u$, $\ID(Q_u)$.
That is, the extended ID of vertex $u$ is given by
\[
\EID(u) = \left[ \ID(u), \LCALabel_{T}(u), \ID(\heavy(u)), \LCALabel_{T}(\heavy(u)), \nlights(u), \ID(Q_u) \right].
\]
This gives the following:

\begin{lemma}\label{lem:bitstrings}
	Let $v$ be a vertex.
	Assume each $w \in I^{\uparrow} (v)$ is associated with a bit $b(w) \in \{0,1\}$ (which may also depend on $v$).
	For $k = |I^{\uparrow} (v)|$, let $S$ be the $k$-bit string whose $i^{\text{th}}$ bit is $b(w)$, where $w$ is the $i^{\text{th}}$ highest vertex in $I^{\uparrow} (v)$, for all $i$.
	Then, given the string $S$ and $\EID(u)$ of any $u \in I^{\uparrow} (v)$, one can report the bit $b(u)$.
\end{lemma}

\begin{proof}
	There are exactly $\nlights(u)$ vertices higher than $u$ in $I^{\uparrow} (v)$, given by $\parent(w')$ for each light $w' \in T[s,u]$.
	Thus $b(u)$ is the $(\nlights(u) + 1)^\text{th}$ bit of $S$.
\end{proof}

\paragraph{Single-Source 1-VFT Labels.}
In Section \ref{sec:single-fault} we presented $1$-VFT connectivity labels of size $O(\log^2 n)$.
If we restrict ourselves to \emph{single-source}, we can modify them to get $O(\log n)$ size, as follows:

The label of vertex $a$ holds the bitstring composed of $\conn(s, b', \Gsub{b})$ for every $b' \in I(a)$ with $\parent(b') = b$ from highest to lowest, and additionally the bit $\conn(s, \heavy(a), \Gsub{a})$.

When decoding, given the labels of a target $t$ and a failure $x$, we first check if $x$ is a heavy or light ancestor of $t$.
If it is neither, we determine that $s,t$ are $x$-connected.
If it a heavy ancestor, we report $\conn(s,\heavy(x),\Gsub{x})$ from the label of $x$.
If it is a light ancestor, we report the $(\nlights(x) + 1)^{\text{th}}$ bit of the bitstring part of the label of $t$. (This is essentially using Lemma \ref{lem:bitstrings}.)

To summarize, we have:
\begin{lemma}\label{lem:ss1vft}
	There is a \emph{single-source} $1$-VFT connectivity labeling scheme with $O(\log n)$-bit labels $\FTLabel{SS1F}(a)$ for every $a \in V$.
	That is, given a query of vertices $\langle t,x \rangle$ along with their $\FTLabel{SS1F}$ labels, one can determine whether the (fixed) source $s$ and $t$ are $x$-connected. 
\end{lemma}

\paragraph{Preliminary Assumptions.}
We assume, w.l.o.g., that the following two conditions hold for the query $\langle t,x,y \rangle$:
\begin{enumerate}[(S1)]
	\item $s,t$ are both $x$-connected and $y$-connected. 
	\item $x \in T[s,t]$ and $y \notin T[x,t]$.
\end{enumerate}
To verify (S1) we use the single-source $1$-VFT labels, which consume only $O(\log n)$ bits, hence within our budget.
If (S1) does not hold, then $s,t$ are clearly $xy$-disconnected and we are done.
We enforce condition (S2) using extended IDs, as follows.
If none of the failures $x,y$ lie in $T[s,t]$, then $s,t$ are $xy$-connected and we are done.
Otherwise, by swapping $x$ and $y$ if necessary, we may assume that $x$ is the lowest failure on $T[s,t]$, which gives (S2).
From this point on, let $x'$ be the child of $x$ on the path $T[x,t]$.
By (S2) we get that $x',t$ are $xy$-connected.
Therefore, in the following it is sufficient for us to determine the $xy$-connectivity of $s,x'$.

Condition (S1) serves to justify the existence of several replacement paths and important vertices which we define in our labels.
However, such explicit justifications are omitted for clarity of presentation.
For the same reason, we also omit explanations about how to extract specific pieces of information from a label storing them (e.g. by using Lemma \ref{lem:bitstrings}).
Throughout, we make extensive use of the properties of interesting sets from Lemma \ref{lem:instersting_sets}.

\paragraph{Strategy.}
We handle differently four different cases based on the location of $y$. The first three are cases handling dependent failures, and the last is for independent failures.\footnote{To distinguish between the cases, we use extended IDs.}
\begin{itemize}
	\item The Down case: $y \in T_{x'}$
	\item The Up case: $y \in T[s,x)$
	\item The Side case: $y \in T_x \setminus T_{x'}$
	\item The Independent case: $x$ and $y$ are independent.
\end{itemize}

\subsection{The Down Case: $y \in T_{x'}$}

\paragraph{Labels for Down Case.}
For a vertex $u$, define
\begin{align}
	A_u &= \{v \in T_u \mid \text{$s,v$ are $(T_u^+ \setminus \{v\})$-connected} \}, \label{eq:A_u} \\
	\alpha_u &= \LCA(A_u), \label{eq:alpha_u} \\ 
	\beta_u &=  \arg \max_{v} \{\depth(v) \mid \text{$v \in T[s,\parent(u))$ and $u,v$ are $(T[s,\parent(u)] \setminus \{v\}))$-connected} \}. \label{eq:beta_u}
\end{align}
The labels for the Down case constructed by Algorithm \ref{alg:D-label}.
Using Lemma \ref{lem:instersting_sets} we see that $\FTLabel{D} (u)$ has length $O(\log^2 n)$ bits.
\begin{algorithm}[!h]
	\caption{Construction of label $\FTLabel{D} (u)$ for vertex $u$}\label{alg:D-label}
	\For{each $v' \in I(u)$ with $\parent(v')=v$}{
		\textbf{store} vertices $v,v', \alpha_{v'}, \beta_{v'}$\;
		\textbf{store} the bitstring 			$
		[ \conn(s,v',\Gsub{v,w}), \text{ $\forall w\in I^{\uparrow}(\alpha_{v'})$ from highest to lowest} ]
		$\;
	}
\end{algorithm}

\paragraph{Decoding Algorithm for Down Case.}
The algorithm proceeds by the following steps.

\bigskip\noindent\textbf{Step (1): Detecting $\alpha := \alpha_{x'}$.}
As $x' \in I(t) \cup I(x)$, $\alpha := \alpha_{x'}$ is stored in either $\FTLabel{D} (t)$ or $\FTLabel{D} (x)$. We treat two easy cases:

\smallskip\noindent{\textit{Case 1: $\alpha \notin T_y$.}}
As $\alpha = \LCA(A_{x'})$, there must be some $u \in A_{x'} \setminus T_y$.
By Eq. \eqref{eq:A_u}, $u \in T_{x'}$ and there is an $s$-$u$ path $P$ internally avoiding $T_{x'}^+$.
Now $P \circ T[u,x']$ certifies that $s,x'$ are $xy$-connected, and we are done.

\smallskip\noindent{\textit{Case 2: $y \in I^{\uparrow}(\alpha)$.}}
Then $\conn(s,x',\Gsub{x,y})$ is specified along with $\alpha$, and we are done.

From now on, assume that $\alpha \in T_y$ and $y \notin I^{\uparrow}(\alpha)$.
Hence, the child of $y$ on $T[y, \alpha]$ is $\heavy(y)$.
It is not hard to see, using Eqs. \eqref{eq:A_u} and \eqref{eq:alpha_u}, that $s,\heavy(y)$ are $xy$-connected by a path of the form $P \circ T[u,\heavy(x)]$ for some $u \in A_{x'}$, where $P$ is an $s$-$u$ path internally avoiding $T_{x'}^+$.

\bigskip\noindent\textbf{Step (2): Detecting $\beta := \beta_{\heavy(y)}$.}
Observe that $\beta := \beta_{\heavy(y)}$ in stored in $\FTLabel{D} (y)$ as $h(y) \in I(y)$.
There are two possible cases.

\smallskip\noindent{\textit{Case 1: $\beta \in T(x,y)$.}}
By Eq. \eqref{eq:beta_u}, there is an $\heavy(y)$-$\beta$ path $P$ internally avoiding $T[s,y]$. Now $P \circ T[\beta, x']$ certifies that $\heavy(y),x'$ are $xy$-connected, hence this is also true for $s,x'$.

\smallskip\noindent{\textit  {Case 2: $\beta \in T[s,x]$.}}
See Illustration in Figure \ref{fig:down}.
\begin{figure}
	\centering
	\includegraphics[scale=0.5]{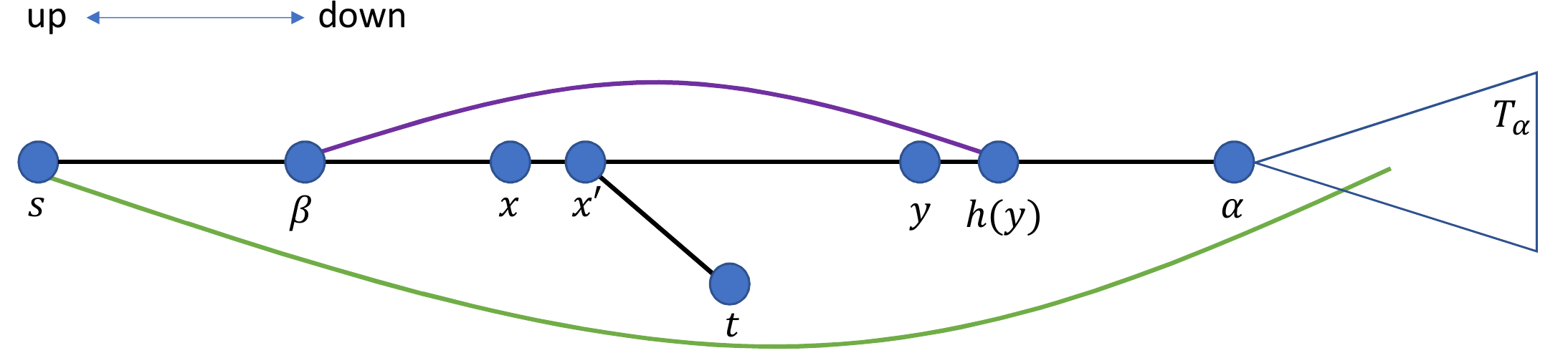}
	\caption{
		The final configuration reached in decoding for the Down case.
		If this configuration is reached then $s,t$ are $xy$-disconnected.
	}
	\label{fig:down}
\end{figure}
We prove that $s,x'$ are $xy$-disconnected.
Assume the contrary, so $P = P_{s,x',\{x,y\}}$ exists.
Let $\widetilde{\alpha}$ be the first vertex from $T_{x'}$ in $P$.
Then $P[s, \widetilde{\alpha}]$ internally avoids $T_{x'}^+$, hence, by Eqs. \eqref{eq:A_u} and \eqref{eq:alpha_u} for $A_{x'}$ and $\alpha = \alpha_{x'}$,
\[
\widetilde{\alpha} \in A_{x'} \subseteq T_{\alpha} \subseteq T_{\heavy(y)} \subseteq T_x.
\]
Now let $\widetilde{\beta}$ be the first vertex from $T[s,y)$ in $P[\widetilde{\alpha}, x']$.
\begin{itemize}
	\item \textit{Case $\widetilde{\beta} \in T[s,x)$:}
	Then $T[s,\widetilde{\beta}]$ avoids $x,y$, so as $P$ is a replacement path going through $s$ and $\widetilde{\beta}$ we obtain
	\[
	\widetilde{\alpha} \in P[s,\widetilde{\beta}] = T[s,\widetilde{\beta}] \subseteq T[s,x),
	\]
	which is a contradiction as $\widetilde{\alpha} \in T_x$.
	
	\item  \textit{Case $\widetilde{\beta} \in T(x,y)$:}
	Then $T[\heavy(y), \widetilde{\alpha}] \circ P[\widetilde{\alpha}, \widetilde{\beta}]$ is an $\heavy(y)$-$\widetilde{\beta}$ path internally avoiding $T[s,y]$,  which contradicts the definition of $\beta = \beta_{\heavy(y)}$ in Eq. \eqref{eq:beta_u} as the lowest vertex having such a path.
\end{itemize}
This concludes the decoding algorithm.

For the sake of optimizing the label length in the subsequent constructions, we observe that one can optimize the length of the labels under if an additional promise is given:
\begin{lemma}\label{lem:down_special}
	There are $O(\log n)$-bit labels $\FTLabel{D'} (u)$ for all vertices $u \in V$ such that given a query of vertices $\langle t, x, y \rangle$ along with their labels, where it is promised that $t,y \in T_{\heavy(x)}$ and $y \notin T[\heavy(x), t]$, one can determine the $xy$-connectivity of $s,t$.
\end{lemma}
\begin{proof}
	We simply replace $I(u)$ by $\{\heavy(u)\}$ in the definition of $\FTLabel{D} (u)$ to get $\FTLabel{D'} (u)$.
	The decoding algorithm is identical to the Down case, only now $x' = \heavy(x) \in I(x)$, which implies that all the required information is stored in the smaller $\FTLabel{D'}$ labels.
\end{proof}

\subsection{The Up Case: $y \in T[s,x)$}

\paragraph{Labels for Up Case.}
For a vertex $u$ with an $s$-$u$ replacement path $P = P_{s,u,\parent(u)}$, let $q_u \in P$ be the first (closest to $s$) vertex in $T_u$.
\begin{observation}\label{obs:the_q_vertex}
	$P_{s,u,\parent(u)} = P' \circ T[q_u, u]$ where $P'$ avoids $T_u$.
\end{observation}
Next, let $Q \in \mathcal{Q}$ be a heavy path such that $u \in Q^{\uparrow}$. We define:
\begin{align}
	a_u &= \arg \min_v \{ \depth(v) \mid v \in T[s,\parent(u)) \text{ and $\exists$ $s$-$u$ path avoiding $T(v,\parent(u)]$}  \}, \label{eq:a_u_first} \\
	b_{u,Q} &= \arg \min_{v} \{\depth(v) \mid \text{$v \in Q^{\uparrow} (u,\ell]$ and $\exists$ $s$-$v$ path internally avoiding $Q^{\uparrow} [u,\ell]$} \}, \label{eq:b_uQ} \\
	c_u &= \arg \max_{v} \{\depth(v) \mid \text{$v \in T[s,u)$ and $\exists$ $\heavy(u)$-$v$ path internally avoiding $T[s,u]$} \}. \label{eq:c_u}
\end{align}
We note that Eq. \eqref{eq:a_u_first} is equivalent to
\begin{align}
	a_u &= \arg \min_v \{ \depth(v) \mid v \in T[s,\parent(u)) \text{ and $\exists$ $v$-$u$ path internally avoiding $T[s,\parent(u)]$} \}. \label{eq:a_u_second}
\end{align}
The labels for the Up case are constructed by Algorithm \ref{alg:U-label}.

$\quad$\\
\begin{algorithm}[H]
	\caption{Construction of label $\FTLabel{U} (u)$ for vertex $u$}\label{alg:U-label}
	\For{each $v' \in I(u)$ with $\parent(v')=v$}{
		\textbf{store} vertices $v,v', a_{v'}$ and value $\conn(s, v', \Gsub{v,a_{v'}})$\;
		\textbf{store} the bitstring 			$
		[ \conn(s,v',\Gsub{v,w}), \text{ $\forall w\in I^{\uparrow}(v')$ from highest to lowest} ]
		$\;
	}
	\textbf{store} vertices $c_u, q_{\heavy(u)}$\;
	\For{each heavy path $Q \in \mathcal{Q}$ intersecting $T[s, q_{\heavy(u)}]$}{
		\textbf{store} $\ID(Q)$, $b_{u,Q}$ and $\conn(s, \heavy(u), \Gsub{u, b_{u,Q}})$\;
	}
\end{algorithm}
$\quad$\\

Using Lemma \ref{lem:instersting_sets} and Observation \ref{obs:heavy_paths} we see that $\FTLabel{U} (u)$ has length $O(\log^2 n)$ bits.

\paragraph{Decoding Algorithm for Up Case.}
The algorithm proceeds by the following steps.

\bigskip\noindent\textbf{Step (0):}
We first handle an easy case:

\smallskip\noindent{\textit{Case: $y \in I^{\uparrow}(x)$.}}
Then as $x' \in I(t) \cup I(x)$, the value $\conn(s,x',\Gsub{x,y})$ is stored in either $\FTLabel{U} (t)$ or $\FTLabel{U} (x)$ (to see this, set $v=x$, $v'=x'$ and $w=y$), so we are done.

From now on, assume that $y \notin I^{\uparrow} (x)$. Thus, the child of $y$ on $T[y,x]$ is $\heavy(y)$.

\bigskip\noindent\textbf{Step (1): Detecting $a := a_{x'}$.}
As $x' \in I(t) \cup I(x)$, $a := a_{x'}$ is stored in either $\FTLabel{U} (t)$ or $\FTLabel{U} (x)$.
We treat two easy cases:

\smallskip\noindent{\textit{Case 1: $a \in T[s,y)$.}}
Then by Eq. \eqref{eq:a_u_first}, there is an $s$-$x'$ path avoiding $T(a,x]$, and hence avoiding $x,y$. This shows that $s,x'$ are $xy$-connected, so we are done.

\smallskip\noindent{\textit{Case 2: $a = y$.}}
Then $\conn(s,x',\Gsub{x,y})$ is specified along with $a$, and we are done.

From now on, assume that $a \in T(y,x) = Q_x^{\uparrow} (y,x)$. By Eq. \eqref{eq:a_u_second}, there is an $a$-$x'$ path $P$ internally avoiding $T[s,x]$, and hence avoiding $x,y$. Thus, $T[\heavy(y), a] \circ P$ shows that $x',\heavy(y)$ are $xy$-connected.

\bigskip\noindent\textbf{Step (2): Detecting $q := q_{\heavy(y)}$.}
Observe that $q := q_{\heavy(y)}$ is stored in $\FTLabel{U} (y)$.
We treat an easy case:

\smallskip\noindent{\textit{Case: $x \notin T[s, q]$.}}
By Observation \ref{obs:the_q_vertex}, $P_{s,\heavy(y),y} = P' \circ T[q, \heavy(y)]$ where $P'$ avoids $T_{\heavy(y)}$, and hence avoids $x$.
Therefore, $P_{s,\heavy(y),y}$ is $s$-$\heavy(y)$ path that avoids both $x,y$.
Thus $s,\heavy(y)$ are $xy$-connected, hence this is also true for $s,x'$, and we are done.

From now on, assume that $x \in T[s, q]$.

\bigskip\noindent\textbf{Step (3): Detecting $b := b_{y, Q_x}$.}
As $x \in T[s, q]$, the heavy path $Q_x$ intersects $T[s,q]$, thus $b := b_{y,Q_x}$ is specified in $\FTLabel{U} (y)$.
Denote by $\ell$ the leaf which is the endpoint of $Q_x$.
We treat two easy cases:

\smallskip\noindent{\textit{Case 1: $b \in T(y,x) = Q_x^{\uparrow} (y,x)$.}}
By Eq. \eqref{eq:b_uQ}, there is an $s$-$b$ path $P$ internally avoiding $Q_x^{\uparrow} [y, \ell]$.
Then $P \circ T[b, \heavy(y)]$ certifies that $s,\heavy(y)$ are $xy$-connected, hence this is also true for $s,x'$, and we are done.

\smallskip\noindent{\textit{Case 2: $b = x$.}}
Then $\conn(s,\heavy(y), \Gsub{x,y})$ is specified with $b$, so we are done again.

From now on, assume that $b \in Q_x^{\uparrow} (x,\ell] = T[\heavy(x), \ell]$.

\bigskip\noindent\textbf{Step (4): Detecting $c := c_x$.}
Observe that $c := c_x$ is stored in $\FTLabel{U} (x)$.
There are two possible cases:

\smallskip\noindent{\textit{Case 1: $c = T(y,x) = Q_x^{\uparrow} (y,x)$.}}
By Eq. \eqref{eq:c_u}, there is an $\heavy(x)$-$c$ path $P$ internally avoiding $T[s,x]$, and hence avoiding $x,y$.
By Eq. \eqref{eq:b_uQ}, there is an $s$-$b$ path $P'$ internally avoiding $Q_x^{\uparrow}[y,\ell]$, and hence avoiding $x,y$.
The concatenation
$
P' \circ T[b,\heavy(x)] \circ P \circ T[c, \heavy(y)]
$
now certifies that $s,\heavy(y)$ are $xy$-connected, hence this is also true for $s,x'$, and we are done.

\smallskip\noindent{\textit{Case 2: $c \in T[s,y] = Q_x^{\uparrow} [s,y]$.}}
See illustration in Figure \ref{fig:up}.
\begin{figure}
	\centering
	\includegraphics[scale=0.5]{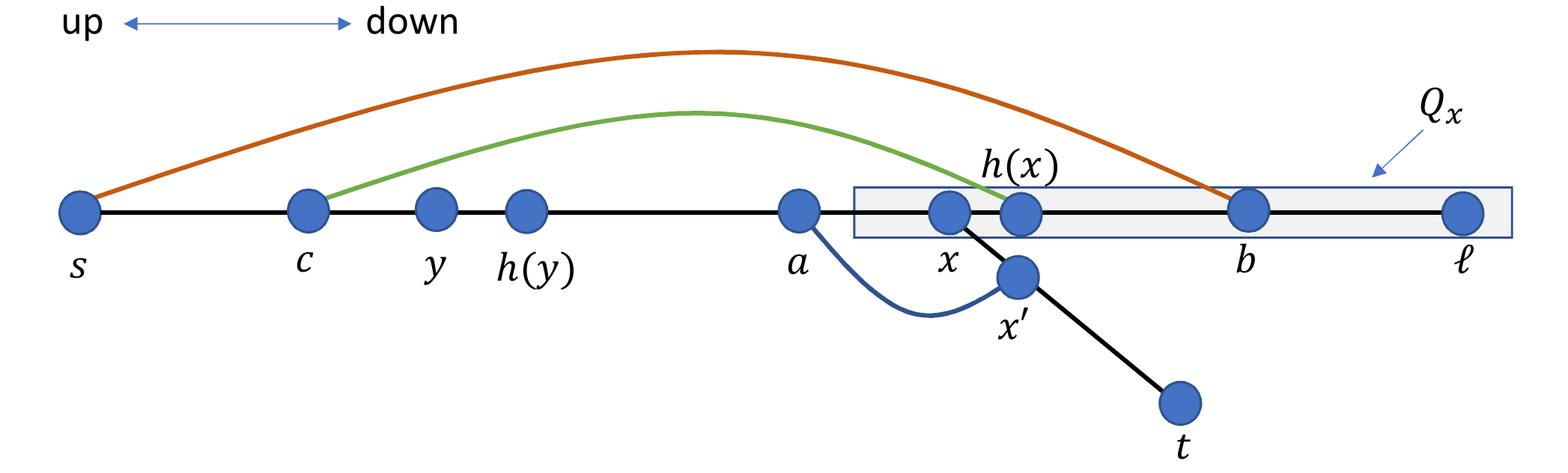}
	\caption{
		The final configuration reached in decoding for the Up case.
		If this configuration is reached then $s,t$ are $xy$-disconnected.
	}
	\label{fig:up}
\end{figure}
We prove that $s,x'$ are $xy$-disconnected.
Assume the contrary, so $P = P_{s,x', \{x,y\}}$ exists.
We divide to cases:
\begin{enumerate}[(a)]
	\item \textit{$P$ avoids $T(y,x) = Q_x^{\uparrow} (y,x)$:}
	Let $\widetilde{a}$ be the lowest vertex from $T[s,y)$ appearing in $P$.
	Then $P$ is an $s$-$x'$ path avoiding $T(\widetilde{a}, x]$, and $\widetilde{a}$ is higher than $a$ --- contradicting Eq. \eqref{eq:a_u_first} for $a = a_{x'}$.
	
	\item \textit{$P$ avoids $T(x,\ell] = Q_x^{\uparrow} (x,\ell]$:}
	As (a) yields a contradiction, $P$ must intersect $Q_x^{\uparrow}(y,x)$.
	Let $\widetilde{b} \in P$ be the first (closest to $s$) vertex in $Q_x^{\uparrow}(y,x)$.
	Then $P[s, \widetilde{b}]$ is an $s$-$\widetilde{b}$ path internally avoiding $Q_x^{\uparrow} [y, \ell]$, and $\widetilde{b}$ is higher than $b$ --- contradicting Eq. \eqref{eq:b_uQ} for $b = b_{y, Q_x}$.
	
	\item \textit{$P$ intersects $T(x,\ell] = Q_x^{\uparrow} (x,\ell]$:}
	Again, as (a) yields a contradiction, $P$ must intersect $Q_x^{\uparrow}(y,x)$.
	Let $\widetilde{c} \in P$ be the first (closest to $s$) vertex in $Q_x^{\uparrow}(y,x) = T(y,x)$.
	It cannot be that $P[s,\widetilde{c}]$ avoids $Q_x^{\uparrow} (x,\ell]$, as this yields the same contradictions as (b).
	Thus, there is a vertex $u \in Q_x^{\uparrow} (x,\ell] = T[\heavy(x), \ell]$ that precedes $\widetilde{c}$ on $P$. 
	As $P$ is a replacement path for faults $x,y$ and $T[s,y)$ is fault-free, $P$ starts by going down a segment of $T[s,y)$ and never returns to $T[s,y)$ again.
	Therefore, the first vertex from $T[s,x]$ that $P[u,\widetilde{c}]$ hits is $\widetilde{c}$.
	Thus, $T[\heavy(x), u] \circ P[u, \widetilde{c}]$ is an $\heavy(x)$-$\widetilde{c}$ path internally avoiding $T[s,x]$.
	Also, $\widetilde{c}$ is lower than $c$ --- contradicting Eq. \eqref{eq:c_u} for $c = c_x$. 
\end{enumerate}
This concludes the decoding algorithm.

Finally, we observe that if extra properties are promised to hold, one can optimize the labeling scheme for the Up case as follows:
\begin{lemma}\label{lem:up_special}
	There are $O(\log n)$-bit labels $\FTLabel{U'} (u)$ for all vertices $u \in V$ such that given a query of vertices $\langle t, x, y \rangle$ along with their labels, where it is promised that $t \in T_{\heavy(x)}$ and $x \notin T_\heavy(y)$, one can determine the $xy$-connectivity of $s,t$.
\end{lemma}
\begin{proof}
	The proof is similar to Lemma \ref{lem:down_special}. Namely, we replace $I(u)$ by $\{\heavy(u)\}$ in the definition of $\FTLabel{U} (u)$ to get $\FTLabel{U'} (u)$, and can still run the  decoding algorithm for the Up case using these labels.
\end{proof}

\subsection{The Side Case: $y \in T_x \setminus T_{x'}$}

\paragraph{Labels for Side Case.}
For a vertex $u$ with a replacement path $P = P_{s,u,\parent(u)}$, let $g_u \in P$ as the last (closest to $u$) vertex in $T_{\heavy(\parent(u))}$.
In case $P$ does not intersect $T_{\heavy(\parent(u))}$, we define $g_u = \Null$.
\begin{observation}\label{obs:the_g_vertex}
	If $g \neq \Null$, then $P_{s,u,\parent(u)}[g_u, u]$ internally avoids $T_{\heavy(\parent(u))}^+$.
\end{observation}
For a subgraph $G' \subseteq G$, denote by $\FTLabel{SS1F} (u, G')$ the single-source $1$-VFT label of $u$ from Lemma \ref{lem:ss1vft} constructed w.r.t. the graph $G'$ with the same source $s$.

The labels for the Side case are constructed by Algorithm \ref{alg:S-label}.

$\quad$\\
\begin{algorithm}[H]
	\caption{Construction of label $\FTLabel{S} (u)$ for vertex $u$}\label{alg:S-label}
	\textbf{store} $\FTLabel{D'}(u)$, $\FTLabel{U'}(u)$\;
	\For{each $v' \in I(u)$ with $\parent(v')=v$}{
		\textbf{store} $\FTLabel{SS1F}(u, \Gsub{v})$, $\FTLabel{SS1F}(v', \Gsub{v})$\;
		\textbf{store} vertex $g_{v'}$, $\conn(s, v', \Gsub{v,g_{v'}})$, $\FTLabel{D'}(g_{v'})$, $\FTLabel{U'}(g_{v'})$\;
		\textbf{store} the bitstring
		$
		[ \conn(s,g_{v'},\Gsub{v,w}), \text{ $\forall w\in I^{\uparrow}(g_{v'})$ from highest to lowest} ]
		$\;
	}
\end{algorithm}

Using Lemma \ref{lem:instersting_sets}, Lemma \ref{lem:ss1vft}, Lemma \ref{lem:down_special} and Lemma \ref{lem:up_special}, we see that  $\FTLabel{S} (u)$ has length $O(\log^2 n)$ bits.

\paragraph{Decoding Algorithm for Side Case.}
The algorithm proceeds by the following steps.

\bigskip\noindent\textbf{Step (0):}
Let $x''$ be this child of $x$ on $T[x,y]$. Then $x'' \neq x'$. We first handle an easy case:

\smallskip\noindent{\textit{Case: $x''$ is light.}}
Then $x'' \in I(y)$, hence $\FTLabel{SS1F} (y, \Gsub{x})$ is stored in $\FTLabel{S} (y)$.
As $x' \in I(t) \cup I(x)$, $\FTLabel{SS1F} (x', \Gsub{x})$ is stored in either $\FTLabel{S} (t)$ or $\FTLabel{S} (x)$.
Using these $\FTLabel{SS1F}$ labels w.r.t. $\Gsub{x}$, we determine if $s,x'$ are connected in $(\Gsub{x}) \setminus \{y\} = \Gsub{x,y}$, and we are done.

From now on assume that $x'' = \heavy(x)$, namely $y \in T_{\heavy(x)}$.

\bigskip\noindent\textbf{Step (1): Detecting $g := g_{x'}$.}
As $x' \in I(t) \cup I(x)$, $g := g_{x'}$ is stored in either $\FTLabel{S} (t)$ or $\FTLabel{S} (x)$.
If $g = \Null$ then the replacement path $P_{s,x',x}$ avoids $T_{\heavy(x)}$, and particularly avoids $y$.
Thus, $s,x'$ are $xy$-connected and we are done.
If $g = y$, then $\conn(s,x', \Gsub{x,y})$ is specified with $g$, so we are done again.

From now on, assume that $g \in V \setminus\{ y \}$.
It follows from Observation \ref{obs:the_g_vertex} that $g,x'$ are $xy$-connected.
Thus, it suffices to determine the $xy$-connectivity of $s,g$.
Observe that $\FTLabel{D'} (g), \FTLabel{U'} (g)$ are specified with $g$.
Also, $\FTLabel{D'} (x), \FTLabel{U'} (x)$ and $\FTLabel{D'} (y), \FTLabel{U'} (y)$ are stored in $\FTLabel{S} (x)$ and $\FTLabel{S} (y)$ respectively.
There are two possible cases (see illustration in Figure \ref{fig:side}):
\begin{figure}
	\centering
	\includegraphics[scale=0.5]{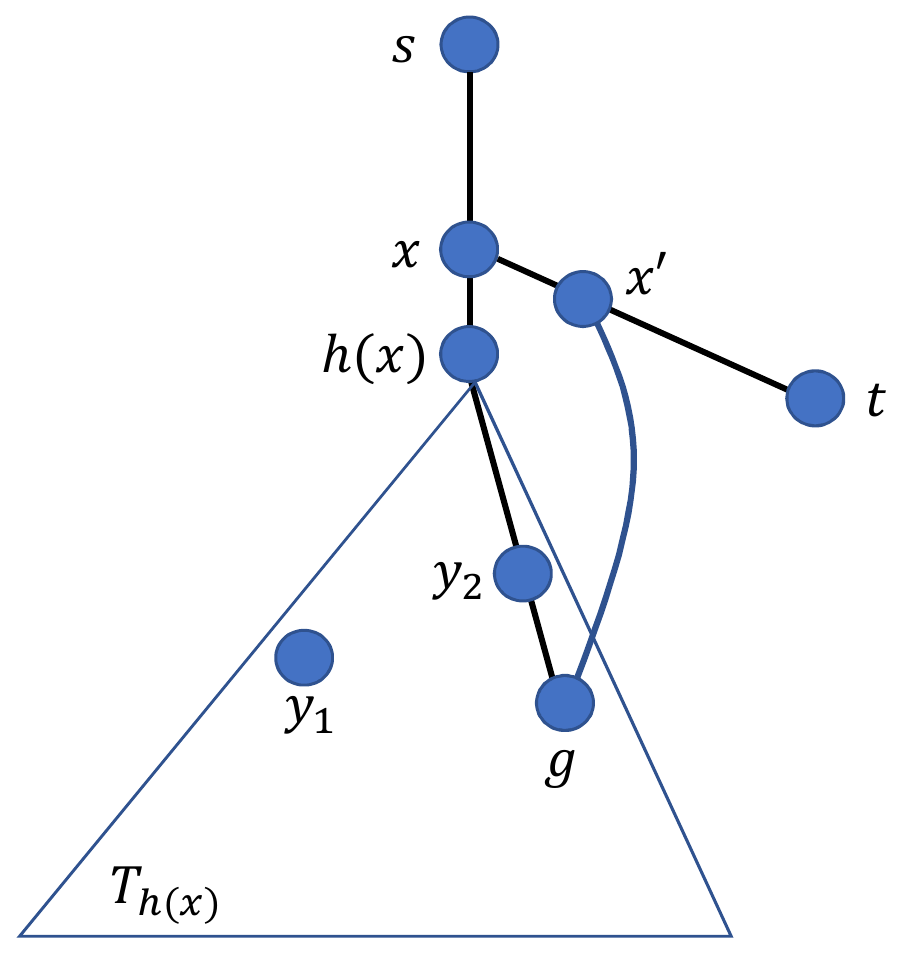}
	\caption{
		The two possible final configurations reached in decoding for the Side case. Vertices $y_1$ and $y_2$ represent the different options for the location of the failure $y$.
	}	
	\label{fig:side}
\end{figure}

\smallskip\noindent{\textit{Case 1: $y \notin T[s,g]$.}
Then $g,y \in T_{\heavy(x)}$ and $y \notin T[\heavy(x),g]$. Thus we apply Lemma \ref{lem:down_special} with query $\langle g, x, y \rangle$ and the corresponding $\FTLabel{D'}$ labels to determine the $xy$-connectivity of $s,g$.

\smallskip\noindent{\textit{Case 2: $y \in T[s,g)$.}
If $y \in I^{\uparrow}(g)$, then $\conn(s,g,\Gsub{x,y})$ is specified along with $g$ and we are done.
Otherwise, the child of $y$ on $T[s,g]$ is $\heavy(y)$, namely $g \in T_{\heavy(y)}$.
Recall that $y \in T_{\heavy(x)}$.
Thus we apply Lemma \ref{lem:up_special} with query $\langle g, y, x \rangle$ and the corresponding $\FTLabel{U'}$ labels to determine the $xy$-connectivity of $s,g$.

\subsection{The Independent Case: $x$ and $y$ are Independent}\label{sec:single-source-dual-ind}

\paragraph{Labels for Independent Case.}
Let $u$ be a vertex. We will need the definition of $\ell_u$ and Observation \ref{obs:the_ell_vertex} from Section \ref{sec:dual-ind}.
Recall that for the replacement path $P = P_{s,u,\parent(u)}$, the vertex $\ell_u \in P$ is the last (closest to $u$) vertex in $T\setminus T_{\parent(u)}$, and Observation \ref{obs:the_ell_vertex} states that $P = T[s, \ell_u] \circ P'$ where $P' \subseteq T_{\parent(u)}$.
Next, let $Q \in \mathcal{Q}$ be a heavy path. We define:
\begin{align}
	d_{u,Q} = \arg \max_{v} \{\depth(v) \mid \text{$v \in Q$ and $\exists$ $s$-$v$ path internally avoiding $\{u\} \cup (Q \setminus T[s,u) )$} \} \label{eq:d_uQ}
\end{align}
The labels for the Independent case are:

$\quad$\\
\begin{algorithm}[H]
	\caption{Construction of label $\FTLabel{IN} (u)$ for vertex $u$}\label{alg:IN-label}
	\For{each $v' \in I(u)$ with $\parent(v')=v$}{
		\textbf{store} $v,v',\ell_{v'}$\;
		\textbf{store} the bitstring
		$
		[ \conn(s,v',\Gsub{v,w}), \text{ $\forall w\in I^{\uparrow}(\ell_{v'})$ from highest to lowest} ]
		$\;
	}
	\For{each heavy path $Q \in \mathcal{Q}$ intersecting $T[s,\ell_{\heavy(u)}]$}{
		\textbf{store} $\ID(Q)$ and $d_{u,Q}$\;
	}
\end{algorithm}

$\quad$\\

Using Lemma \ref{lem:instersting_sets} and Observation \ref{obs:heavy_paths} we see that $\FTLabel{IN} (u)$ has length $O(\log^2 n)$ bits.

\paragraph{Decoding Algorithm for Independent Case.}
The algorithm proceeds by the following steps.

\bigskip\noindent\textbf{Step (1): Detecting $\ell_{x}$.}
As $x' \in I(t) \cup I(x)$, $\ell_{x'}$ is stored in either $\FTLabel{IN} (t)$ or $\FTLabel{IN} (x)$.
We handle two easy cases:

\smallskip\noindent{\textit{Case 1: $y \notin T[s,\ell_{x'}]$.}}
By Observation \ref{obs:the_ell_vertex}, it holds that $P_{s,x',x} = T[s,\ell_{x'}] \circ P'$ where $P' \subseteq T_x$. Therefore, $P_{s,x',x}$ also avoids $y$, so $s,x'$ are $xy$-connected and we are done.

\smallskip\noindent{\textit{Case 2: $y \in I^{\uparrow}(\ell_{x'})$.}}
Then $\conn(s,x',\Gsub{x,y})$ is specified next to $\ell_{x'}$, so we are done.

From now on assume that $y \in T[s,\ell_{x'}]$ and $y \notin I^{\uparrow} (\ell_{x'})$.
Thus, the child of $y$ on $T[y,\ell_{x'}]$ is $\heavy(y)$.
Now $P_{s,x',x} [\heavy(y), x']$ shows that $x', \heavy(y)$ are $xy$-connected.

\bigskip\noindent\textbf{Step (2): Detecting $\ell_{\heavy(y)}$ and $\ell_{\heavy(x)}$.}
As $\heavy(y) \in I(y)$, $\ell_{\heavy(y)}$ is stored in $\FTLabel{IN} (y)$.
The cases where $x \notin T[s, \ell_{\heavy(y)}]$ and where $x \in I^{\uparrow}(\ell_{\heavy(y)})$  are treated similarly to the symmetric cases for $\ell_{x'}$.
Hence, in these cases we can determine the $xy$-connectivity of $s,\heavy(y)$, and thus also of $s,x'$.

From now on assume that $x \in T[s, \ell_{\heavy(y)}]$ and $x \notin I^{\uparrow} (\ell_{\heavy(y)})$.
Thus, the child of $x$ on $T[x, \ell_{\heavy(y)}]$ is $\heavy(x)$.
Now $P_{s,\heavy(y), y}[\heavy(x), \heavy(y)]$ shows that $\heavy(y), \heavy(x)$ are $xy$-connected.

In a symmetric manner, we find $\ell_{\heavy(x)}$ in $\FTLabel{IN} (x)$.
We act in a similar fashion as with $\ell_{\heavy(y)}$.
The only case where we do not finish is when $y \in T[s, \ell_{\heavy(x)}]$ and $y \notin I^{\uparrow} (\ell_{\heavy(x)})$, so from now on we assume this case holds.

\bigskip\noindent\textbf{Step (3): Inspecting $Q_x$ and $Q_y$.}
Let $b_x,b_y$ be the leaves of $T$ which are the endpoints of $Q_x, Q_y$ respectively.
Every $w \in Q_x[\heavy(x), b_x]$ is $xy$-connected to $\heavy(x)$, and every $w \in Q_y[\heavy(y), b_y]$ is $xy$-connected to $\heavy(y)$.
Recall now that $x', \heavy(x), \heavy(y)$ are all $xy$-connected.
Let $B = Q_x[\heavy(x), b_x] \cup  Q_y[\heavy(y), b_y]$ (i.e., all vertices below $x,y$ in $Q_x,Q_y$ respectively).
Then all vertices in $B$ are $xy$-connected, and our goal becomes determining whether there exists $w \in B$ that is $xy$-connected to $s$.
Next, let $A = T[s,x] \cup T[s,y]$. We claim that all vertices in $B$ are not only $xy$-connected, but in fact $A$-connected. Indeed, the path $P_{s,\heavy(y),y} [\heavy(x), \heavy(y)]$ serves as a bridge avoiding $A$ between the parts $Q_x [\heavy(x), b_x]$ and $Q_y [\heavy(y), b_y]$ of $B$.

\bigskip\noindent\textbf{Step (4): Detecting $d_x := d_{x,Q_y}$ and $d_y := d_{y,Q_x}$.}
As $y \in T[s, \ell_{\heavy(x)}]$, $Q_y$ intersects $T[s, \ell_{\heavy(x)}]$, so $d_x := d_{x,Q_y}$ is stored in $\FTLabel{IN} (x)$.
We also find $d_y := d_{y, Q_x}$ symmetrically.
There are two possible cases:

\smallskip\noindent{\textit{Case 1: $d_x \in B$ or $d_y \in B$.}}
Assume $d_x \in B$ (the case $d_y \in B$ is similar).
By Eq. \eqref{eq:d_uQ}, there exists an $s$-$d_x$ path $P$ internally avoiding $\{x\} \cup (Q_y \setminus T[s,x) )$, hence $P$ avoids $x,y$.
Thus, $s$ and $d_x \in B$ are $xy$-connected, so we are done.

\smallskip\noindent{\textit{Case 2: $d_x, d_y \in A$.}}
See illustration in Figure \ref{fig:ind_ss}.
\begin{figure}
	\centering
	\includegraphics[scale=0.5]{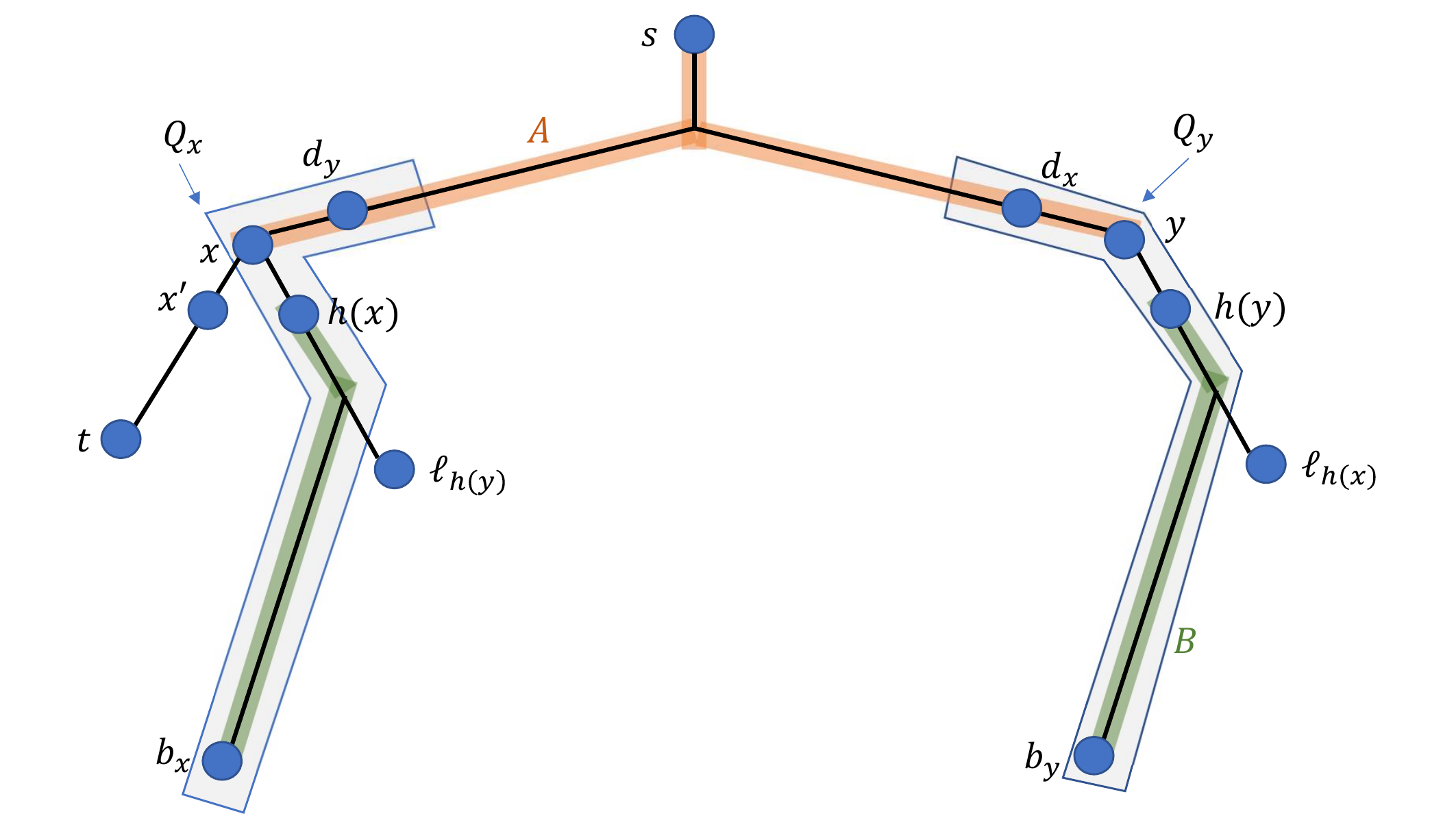}
	\caption{
		The final configuration reached in decoding for the Independent case.
		If this configuration is reached then $s,t$ are $xy$-disconnected.
	}
	\label{fig:ind_ss}
\end{figure}
We prove that all vertices in $B$ are $xy$-disconnected from $s$, which concludes the algorithm.
Assume the contrary, and let $P$ be an $s$-$w$ path avoiding $x,y$ for some $w \in B$.
We may assume that $w$ is the only vertex from $B$ in $P$ (otherwise, replace $w$ by the first hitting point in $B$, and trim $P$ to end there).
Let $u \in P$ be the last (closest to $w$) vertex from $A$.
Assume $u \in T[s,x)$ (the case $u \in T[s,y)$ is similar).
Then $T[s,u] \circ P[u,w]$ internally avoids $\{x\} \cup (Q_y \setminus T[s,x))$.
We divide to cases:
\begin{itemize}
	\item $w \in Q_y[\heavy(y), b_y]$:
	Then $w$ is lower than $d_x$, contradicting Eq. \eqref{eq:d_uQ} for $d_x = d_{x,Q_y}$.
	\item $w \in Q_x [\heavy(x), b_x]$:
	As all vertices in $B$ are $A$-connected and $\heavy(y) \in B$, there exists a $w$-$\heavy(y)$ path $P'$ internally avoiding $A$.
	Let $w'$ be the first vertex from $Q_y$ that $P'$ hits.
	Then $T[s,u] \circ P[u,w] \circ P'[w,w']$ internally avoids $\{x\} \cup (Q_y \setminus T[s,x))$. But as $w' \in Q_y \setminus A$, it is lower than $d_x$, which again contradicts Eq. \eqref{eq:d_uQ} for $d_x = d_{x,Q_y}$.
\end{itemize}
This concludes the decoding label for the Independent case.

Finally, by concatenating the labels $\FTLabel{D} (u)$, $\FTLabel{U} (u)$, $\FTLabel{S} (u)$ and $\FTLabel{IN} (u)$ labels for each $u \in V$, we get a labeling scheme proving Lemma \ref{lem:orig_ss-2vft_labels}.

\end{document}